\documentclass[11pt]{article}

\usepackage{beton}
\usepackage{centernot}

\usepackage{amsmath}
\usepackage{amssymb}
\usepackage{amsthm}
\usepackage[dvipsnames]{xcolor}
\usepackage{mathtools}
\usepackage{amsfonts}
\usepackage{graphicx}
\usepackage{mathrsfs}
\usepackage{braket}
\usepackage{soul}
\usepackage{enumitem}
\usepackage{bm} 
\usepackage[ruled,linesnumbered]{algorithm2e}
\usepackage[T1]{fontenc}
\usepackage{fullpage}
\usepackage[colorlinks,linkcolor=Maroon,citecolor=DarkBlue]{hyperref}
\usepackage[capitalize,nameinlink]{cleveref}
  \crefname{step}{Step}{Steps}
\usepackage[hyperpageref]{backref}

\definecolor{DarkBlue}{RGB}{0,0,150}
\definecolor{DarkRed}{RGB}{150,0,0}
\definecolor{DarkGreen}{RGB}{0,150,0}

\newtheorem{theorem}{Theorem}[section]
\newtheorem{proposition}[theorem]{Proposition}

\newtheorem{lemma}[theorem]{Lemma}

\newtheorem{definition}[theorem]{Definition}

\newtheorem{remark}[theorem]{Remark}
\newtheorem{example}[theorem]{Example}

\theoremstyle{definition}

\newcommand{\q}{q}

\newcommand{\U}{U} 

\newcommand{\aux}{j}

\def\cA{{\cal A}}
\def\cB{{\cal B}}

\def\cH{{\cal H}}
\def\cI{{\cal I}}

\def\cK{{\cal K}}

\def\cO{{\cal O}}

\newcommand{\Id}{\mathsf{Id}}

\newcommand{\C}{\mathbb{C}}
\newcommand{\R}{\mathbb{R}}
\newcommand{\N}{\mathbb{N}}

\newcommand{\HH}{\mathbb{H}}

\newcommand{\re}{\operatorname{re}}

\newcommand{\opm}{\operatorname{M}}

\newcommand*{\mtxc}[1]{\opm_{#1}(\C)}
\newcommand*{\mtxr}[1]{\opm_{#1}(\R)}
\newcommand*{\mtxh}[1]{\opm_{#1}(\HH)}

\DeclareMathOperator{\ran}{ran}
\DeclareMathOperator{\Alg}{Alg}

\newcommand{\tr}{\operatorname{Tr}}
\newcommand{\supp}{\operatorname{supp}}

\newcommand{\proj}[1]{| #1\rangle\!\langle #1 |}

\newcommand{\ie}{\textit{i}.\textit{e}.}
\newcommand{\eg}{\textit{e}.\textit{g}.}

\title{
Beyond real: Investigating the role of complex numbers in self-testing} 

\author{Ranyiliu Chen \thanks{\scriptsize{Quantum Science Center of Guangdong-HongKong-Macao Greater Bay Area.} \footnotesize E-mail: \texttt{chenranyiliu@quantumsc.cn}.} \and Laura Man\v{c}inska \thanks{QMATH, Department of
Mathematical Sciences, {University of Copenhagen}. E-mail: \texttt{mancinska@math.ku.dk}.} \and Jurij Vol\v{c}i\v{c} \thanks{Department of Mathematics, University of Auckland. E-mail: \texttt{jurij.volcic@auckland.ac.nz}.}}
\date{\today}

\begin{document}
\maketitle

\begin{abstract} 
We investigate complex self-testing, a generalization of standard self-testing that accounts for quantum strategies whose statistics is indistinguishable from their complex conjugate's. We show that many structural results from standard self-testing extend to the complex setting, including lifting of common assumptions. Our main result is an operator-algebraic characterization: complex self-testing is equivalent to uniqueness of the real parts of higher moments, leading to a basis-independent formulation in terms of real C* algebras. This leads to a classification of non-local strategies, and a tight boundary where standard self-testing does not apply and complex self-testing is necessary. We further construct a strategy involving quaternions, establishing the first standard self-test for genuinely complex strategy. Our work clarifies the structure of complex self-testing and highlights the subtle role of complex numbers in bipartite Bell non-locality.
\end{abstract}

\vfill

\thispagestyle{empty}
\newpage
\setcounter{tocdepth}{2}
\tableofcontents
\thispagestyle{empty}

\clearpage

\section{Introduction}\label{sec:intro}

\emph{Self-testing} is a powerful concept in quantum information theory that enables the certification of quantum states and measurements solely from observed correlations. This idea originates from Bell non-locality \cite{Bell1964paradox,CHSH1969}. In the 1980s, studies identified the maximal quantum violation of the CHSH inequality, showing that it is uniquely attained by a specific entangled state and measurement setup \cite{Tsirelson1987QuantumAO,POPESCU1992411}. This uniqueness serves as the foundation for the notion of self-testing, formalized by Mayers and Yao \cite{MY}. Since its inception, self-testing has evolved into an active and expanding field of research. The utility of self-testing extends broadly and it lies behind applications across quantum information science \cite{Andrea2019Verifier,BCM,MY, DIQKD,mipre}. A comprehensive overview can be found in \cite{SB}.

It is a folklore fact that complex strategies cannot be self-tested, at least not in the standard sense \cite[Section 3.7.1]{SB}. To see this, consider a strategy $S=(\ket{\psi},\{E_{xa}\},\{F_{yb}\})$ for a (bipartite) Bell scenario; that is, Alice performs $\{E_{xa}\}$, and Bob performs $\{F_{yb}\}$, on their respective parts of the shared state $\ket\psi$. The resulting correlation satisfies
$$
p(a,b|x,y)=\braket{\psi|E_{xa}\otimes F_{yb}|\psi}=\overline{\braket{\psi|E_{xa}\otimes F_{yb}|\psi}}.
$$
That is, $S$ and $\overline{S}:=(\ket{\overline{\psi}},\{\overline{E_{xa}}\},\{\overline{F_{yb}}\})$ always give rise to a same correlation $p$. Since there is generally no local unitary mapping between $S$ and $\overline{S}$,
\footnote{A simple example is a strategy contains all three Pauli measurements: there is no unitary map that takes $(\sigma_X,\sigma_Y,\sigma_Z)$ to $(\sigma_X,\overline{\sigma_Y}=-\sigma_Y,\sigma_Z)$} 
the standard self-testing framework, which certifies uniqueness up to local unitaries, is an ill fit for such cases. This is especially relevant in multipartite Bell scenarios where quantum states cannot always be represented by real numbers \cite{supic2023,Maria2024puremultipartiteentangledstates}. 

To address this, the notion of \emph{complex self-testing} \cite{MM11Generalized} was introduced, allowing any combination of a strategy and its complex conjugate. This idea has been explored in several works, first in the bipartite scenario \cite{MM11Generalized,AA2016,bowles_self-testing_2018,9062494}, and also in multipartite scenarios \cite{supic2023,Maria2024puremultipartiteentangledstates}. It is worth noting that, certifying complex measurements is tied to deep foundational problems in quantum information theory \cite{renou_quantum_2021}. However, complex self-testing remains relatively underdeveloped in comparison to the rich theory of (standard) self-testing. Many questions regarding its structural properties and limits are yet to be answered.

In this work, we undertake a systematic study of complex self-testing. Our main contributions are as follows:

\begin{itemize}
    \item We show that many fundamental properties of standard self-testing carry over to the complex setting. This generalizes the results of removing common assumptions like projectivity and full-rankness from \cite{baptista2023mathematicalfoundationselftestinglifting}.
    \item We provide an operator-algebraic characterization of complex self-testing. Specifically, we show that complex self-testing is equivalent to uniqueness of the real part of higher moments across all strategies reproducing the same correlation. This leads to a natural basis-independent definition of complex self-testing, in terms of unique real states on real C* algebras--contrasting with the (complex) C* algebra framework used for standard self-testing \cite{paddock2023operatoralgebraic}.
    \item As a consequence, we prove that complex self-testing reduces to standard self-testing whenever the canonical strategy has only real-valued moments. This observation prompts a deeper investigation into the “realness” of quantum strategies and leads us to identify a fundamental boundary: 
    strategies with nonreal higher moments
    cannot be certified via standard self-testing.
    \item Finally, we explore an intermediate regime between 
    real-representable 
    strategies and those with real higher moments.
    We construct a new self-test that relies on quaternions, demonstrating that nontrivial (standard) self-tests exist in this middle ground. A technical by-product of our construction is a tight lower bound on the number of projections required to generate the quaternion matrix algebra. This gives the minimal scenario for this type of self-tested strategies to exist, and might be of independent interest from an operator-algebraic point of view.
\end{itemize}

Through this work, we aim to provide a coherent and comprehensive foundation for complex self-testing, clarify its relationship to standard self-testing, and open new directions for understanding the algebraic structure underlying non-locality.
\section{Preliminaries and notion}\label{sec:pre}

Throughout this paper all Hilbert spaces (denoted by $\cH$, with subscripts indicating the party they belong to) are assumed to be over the complex field and finite-dimensional, unless specified otherwise. The set of linear operators on Hilbert space is denoted by $L(\cH)$. 
The identity operator of a $d$-dimensional Hilbert space is denoted by $\Id$, with subscripts indicating the party. For a matrix $a$, let $a^*,\overline{a},a^\intercal$ denote its adjoint, complex conjugate, and transpose, respectively. A (pure) state (namely, a unit vector) in a Hilbert space $\cH=\C^n$ is denoted by $\ket{\psi}$ and we use $\ket{\overline\psi}$ to denote its complex conjugate. 
Real, complex, and quaternion numbers are denoted by $\R,\C,$ and $\HH$, respectively.

\subsection{Bell scenarios}

In a standard bipartite Bell scenario {\cite{Bell1964paradox,Brunner_2014}}, non-communicating players Alice and Bob are modelled as performing local measurements on their respective part of their shared state. Their behaviour then is described as \emph{quantum strategies}: $$S=\Big(
\ket{\psi}_{AB}\in\mathcal{H}_A\otimes\mathcal{H}_B,\{E_{xa}:{x\in\mathcal{I}_A,a\in\mathcal{O}_A}\}\subset L(\mathcal{H}_A),\{F_{yb}:{y\in\mathcal{I}_B,b\in\mathcal{O}_B}\}\subset L(\mathcal{H}_A)
\Big).$$
In each round of the interaction, a verifier samples inputs $x \in \mathcal{I}_A$ and $y \in \mathcal{I}_B$ from finite question sets $\mathcal{I}_{A,B}$ and sends them to Alice and Bob respectively. In response, players produce outputs $a \in \mathcal{O}_A$ and $b \in \mathcal{O}_B$ from finite answer sets $\mathcal{O}_{A,B}$, based on local quantum measurements $\{E_{xa}\},\{F_{yb}\}$ applied to their respective subsystems. Therefore, the statistics, also called the \emph{correlation} of $S$, is given by the conditional probability distribution $p(a,b|x,y) = \braket{\psi|E_{xa} \otimes F_{yb}|\psi}$, which can be estimated through repeated executions of the protocol. The set of all possible correlations generated by finite dimensional tensor product quantum strategies is denoted by $C_q$.

In \cite{baptista2023mathematicalfoundationselftestinglifting}, basic properties about non-local strategies were introduced.

\begin{definition}
    A strategy $S=(\ket{\psi} \in \mathcal{H}_{A}\otimes \mathcal{H}_B,\{E_{xa}\},\{F_{yb}\})$ is 
    \begin{itemize}
        \item \emph{support-preserving} if
    \begin{align*}
        [\Pi_A,E_{xa}]=[\Pi_B,B_{yb}]=0,
    \end{align*}
    hold for all $x,y,a,b$, where $\Pi_A$ (resp. $\Pi_B$) is the projection onto the support of $\tr_B\proj{\psi}$ (resp. $\tr_A\proj{\psi}$);
    \item \emph{full-rank} if $\ket{\psi}$ has full Schmidt rank;
    \item \emph{0-projective} (or `projective on the state') if
    \begin{align*}
        \braket{\psi|(\Id_A-E_{xa})E_{xa}\otimes\Id_B|\psi}=\braket{\psi|\Id_A\otimes(\Id_B-F_{yb})F_{yb}|\psi}=0,
    \end{align*}
    hold for all $x,y,a,b$;
    \item \emph{projective} if $E_{xa}^2=E_{xa},F_{yb}^2=F_{yb}$ holds for all $a,b,x,y$.
    \end{itemize}
\end{definition}

Clearly, full-rank implies support-preserving, and projective implies 0-projective.

The following property of support-preservingness will be useful in our proofs.
\begin{lemma}[Lemmas 4.3 \& 4.4 of \cite{paddock2023operatoralgebraic}; see also Lemma 3.3 of \cite{baptista2023mathematicalfoundationselftestinglifting} for the approximate version]
    A strategy $(\ket{\psi},\{E_{xa}\},\{F_{yb}\})$ is support-preserving if and only if there exist operators $\hat{E}_{xa}$, $\hat{F}_{yb}$ such that $E_{xa}\otimes \Id\ket{\psi}=\Id\otimes\hat{E}_{xa}\ket{\psi}$ and $\Id\otimes F_{yb}\ket{\psi}=\hat{F}_{yb}\otimes \Id\ket{\psi}$ for all $x,y,a,b$.
    \label{lem:carryover}
\end{lemma}

A strategy $S$ as above is \emph{irreducible} if neither $\{E_{xa}\}$ nor $\{F_{yb}\}$ have a non-trivial proper closed invariant subspace. If $S$ is finite-dimensional (in the sense that $\dim\cH_A,\dim\cH_B<\infty$), this is equivalent to $\{E_{xa}\}$ and $\{F_{yb}\}$ generating $L(\cH_A)$ and $L(\cH_B)$ as complex algebras. 

For the sake of simplicity, we denote words (products of operators) of length $k$ by $E_{\vec{x}\vec{a}}:=E_{x_ka_k}E_{x_{k-1}a_{k-1}}\cdots E_{x_1a_1}, \tilde{E}_{\vec{x}\vec{a}}:=\tilde{E}_{x_ka_k}\tilde{E}_{x_{k-1}a_{k-1}}\cdots \tilde{E}_{x_1a_1}$, where $\vec{x}:=(x_k,\dots, x_1),\vec{a}:=(a_k,\dots, a_1)$. Similarly for Bob's operators, $F_{\vec{y}\vec{b}}:=F_{y_\ell b_\ell }F_{y_{\ell -1}b_{\ell -1}}\cdots F_{y_1b_1}, \tilde{F}_{\vec{y}\vec{b}}:=\tilde{F}_{y_\ell b_\ell }\tilde{F}_{y_{\ell -1}b_{\ell -1}}\cdots \tilde{F}_{y_1b_1}$ denote words of length $\ell$. 
One may view the correlation $\braket{\psi|E_{xa} \otimes F_{yb}|\psi}$ of $S$ as the first moments of the joint distribution of Alice and Bob. In line with this view, we call $\braket{\psi|E_{\vec{x}\vec{a}} \otimes F_{\vec{y}\vec{b}}|\psi}$ the \emph{higher moments} of the strategy $S$.

\subsection{Complex dilation and complex self-testing}

Self-testing aims to establish a correspondence between the \emph{canonical} strategy and the \emph{physical} strategy. The canonical strategy $\tilde S=(\ket{\tilde\psi}_{AB},\{\tilde E_{xa}\},\{\tilde F_{yb}\})$ is the specification or the blueprint to be compared with, while the physical strategy $ S=(\ket{\psi}_{AB},\{E_{xa}\},\{F_{yb}\})$ is performed by the players. The notion of {local dilation}, introduced in \cite{MPS,paddock2023operatoralgebraic} and now standard in the literature, describes such a correspondence (more precisely, a partial order), incorporating undetectable auxiliary resource and change of the frame of reference.

\begin{definition}[Local dilation]
A strategy $\tilde S=(\ket{\tilde\psi},\{ \tilde E_{xa}\},\{\tilde F_{yb}\})$ is a \emph{local dilation} of a strategy $S=(\ket{\psi},\{ E_{xa}\},\{ F_{yb}\})$ (denoted $S \xhookrightarrow{}\tilde{S}$) if there exist a local isometry $U=U_A\otimes U_B$ with 
$$U_A:\mathcal{H}_A\rightarrow\mathcal{H}_{\tilde{A}}\otimes\cH_{\hat A},$$ $$U_B:\mathcal{H}_B\rightarrow\mathcal{H}_{\tilde{B}}\otimes\cH_{\hat B}$$
and an auxiliary state $\ket\aux$ such that

\begin{align*}
U\ket{\psi}_{AB}&=\ket{\tilde{\psi}}_{\tilde{A}\tilde{B}}\ket{\aux}_{\hat{A}\hat{B}}\\
    U(E_{xa}\otimes \Id_B)\ket{\psi}_{AB}&=(\tilde{E}_{xa}\otimes \Id_{\tilde B}\ket{\tilde{\psi}}_{\tilde{A}\tilde{B}})\ket{\aux}_{\hat{A}\hat{B}}\\
    U(\Id_A\otimes F_{yb})\ket{\psi}_{AB}&=(\Id_{\tilde A}\otimes\tilde{F}_{yb}\ket{\tilde{\psi}}_{\tilde{A}\tilde{B}})\ket{\aux}_{\hat{A}\hat{B}}
    \end{align*}
hold for all $a,b,x,y$.
\label{def:localdilation}
\end{definition}

In a \emph{complex} self-test, the physical strategy is expected to be an arbitrary combination of $\tilde S$ and its complex conjugate. We formulate this with a \emph{complex local dilation}, the `complex' analogue of a local dilation. The idea is to introduce an additional Hilbert space $\cH_{A'}\otimes\cH_{B'}$ on which the devices perform measurements on an entangled state to concurrently employ the canonical strategy or its complex conjugate. Since the additional measurements acting on $\cH_{A'}$ and $\cH_{B'}$ has binary outcomes, without loss of generality we take $\cH_{A'}\cong\cH_{B'}\cong\C^2$, and the state in $\cH(A')\otimes\cH(B')$ can take the form $\alpha\ket{00}+\beta\ket{11}$. Also notice that the real coefficients $\alpha,\beta$ can be absorbed to the auxiliary states. We then define complex local dilation as follows.

\begin{definition}[Complex local dilation]
A strategy $\tilde{S}=(\ket{\tilde{\psi}},\{\tilde{E}_{xa}\},\{\tilde{F}_{yb}\})$ is a \emph{complex local dilation} of a strategy ${S}=(\ket{\psi},\{E_{xa}\},\{F_{yb}\})$ (denoted ${S}\xhookrightarrow{}_{\C} \tilde{S}$) if there exists a local isometry $U=U_A\otimes U_B$ with $$U_A:\mathcal{H}_A\rightarrow\mathcal{H}_{\tilde{A}}\otimes\cH_{\hat A}\otimes\mathcal{H}_{A'},$$ $$U_B:\mathcal{H}_B\rightarrow\mathcal{H}_{\tilde{B}}\otimes\cH_{\hat B}\otimes\mathcal{H}_{B'}$$
such that
\begin{align}
U\ket{\psi}_{AB}&=\ket{\tilde{\psi}}_{\tilde{A}\tilde{B}}\ket{\aux_0}_{\hat{A}\hat{B}}\ket{00}_{A'B'}+\ket{\overline{\tilde{\psi}}}_{\tilde{A}\tilde{B}}\ket{\aux_1}_{\hat{A}\hat{B}}\ket{11}_{A'B'},\\
    U(E_{xa}\otimes \Id_B)\ket{\psi}_{AB}&=(\tilde{E}_{xa}\otimes \Id_{\tilde B}\ket{\tilde{\psi}}_{\tilde{A}\tilde{B}})\ket{\aux_0}_{\hat{A}\hat{B}}\ket{00}_{A'B'}+(\overline{\tilde{E}_{xa}}\otimes \Id_{\tilde B}\ket{\overline{\tilde{\psi}}}_{\tilde{A}\tilde{B}})\ket{\aux_1}_{\hat{A}\hat{B}}\ket{11}_{A'B'},\label{eq:combin1A1}\\
    U(\Id_A\otimes F_{yb})\ket{\psi}_{AB}&=(\Id_{\tilde A}\otimes\tilde{F}_{yb}\ket{\tilde{\psi}}_{\tilde{A}\tilde{B}})\ket{\aux_0}_{\hat{A}\hat{B}}\ket{00}_{A'B'}+(\Id_{\tilde A}\otimes\overline{\tilde{F}_{yb}}\ket{\overline{\tilde{\psi}}}_{\tilde{A}\tilde{B}})\ket{\aux_1}_{\hat{A}\hat{B}}\ket{11}_{A'B'}\label{eq:combin1B1}
    \end{align}
%
hold for all $a,b,x,y$, where 
$\ket{\aux_0}$ and $\ket{\aux_1}$ are (not necessarily orthogonal) subnormalized states satisfying $\braket{\aux_0|\aux_0}+\braket{\aux_1|\aux_1}=1$. 
\label{def:complex2}
\end{definition}

Clearly, if $\tilde S$ is already represented by real matrices, $\ket{00}_{A'B'}$ and $\ket{11}_{A'B'}$ can be absorbed into the auxiliary state, in which case the complex local dilation degenerates to a (standard) local dilation. On the other hand, (standard) local dilation always implies complex local dilation by letting $\ket{\aux_1}=0$. We also note that, while it is always possible to take a basis in which the canonical state is real (thanks to the Schmidt decomposition \cite{Nielsen_Chuang_2010}), the above definition does not rely on the assumption of a real matrix representation of the canonical state.

In an alternative definition, a complex local dilation can be also understood as a \emph{convex combination} of $\tilde S$ and its complex conjugate, as introduced in \cite{mancinska2021glued}. If we see local systems as a direct sum of subsystems:
$$
\mathcal{H}_A=\mathcal{H}_{A_0}\oplus\mathcal{H}_{A_1},\ \mathcal{H}_B=\mathcal{H}_{B_0}\oplus\mathcal{H}_{B_1},
$$
then the whole system satisfies
$$
\mathcal{H}_{AB}=\mathcal{H}_{A}\otimes\mathcal{H}_{B}\cong\mathcal{H}_{A_0}\otimes\mathcal{H}_{B_0}\oplus\mathcal{H}_{A_0}\otimes\mathcal{H}_{B_1}\oplus\mathcal{H}_{A_1}\otimes\mathcal{H}_{B_0}\oplus\mathcal{H}_{A_1}\otimes\mathcal{H}_{B_1}.
$$
If we only focus on vectors in the subspace $$\mathcal{H}_{A_0}\otimes\mathcal{H}_{B_0}\oplus\mathcal{H}_{A_1}\otimes\mathcal{H}_{B_1}\subsetneq\mathcal{H}_{A}\otimes\mathcal{H}_{B},$$
we then use the following diagonal direct sum notation for vectors $v_0\in\mathcal{H}_{A_0}\otimes\mathcal{H}_{B_0}, v_1\in\mathcal{H}_{A_1}\otimes\mathcal{H}_{B_1}$:
$$
v_0\oplus_{\Delta}v_1:=v_0\oplus\vec{0}_{\mathcal{H}_{A_0}\otimes\mathcal{H}_{B_1}}\oplus\vec{0}_{\mathcal{H}_{A_1}\otimes\mathcal{H}_{B_0}}\oplus v_1\in\mathcal{H}_{AB}.
$$
That is, $v_0\oplus_{\Delta}v_1$ is a vector in the subspace $\mathcal{H}_{A_0}\otimes\mathcal{H}_{B_0}\oplus\mathcal{H}_{A_1}\otimes\mathcal{H}_{B_1}\subsetneq\mathcal{H}_{AB}$.

\begin{definition}[complex local dilation, alternative]
A strategy $\tilde{S}=(\ket{\tilde{\psi}}_{\tilde{A}\tilde{B}},\{\tilde{E}_{xa}\},\{\tilde{F}_{yb}\})$ is a complex local dilation of ${S}=(\ket{\psi}_{AB},\{E_{xa}\},\{F_{yb}\})$ if there exists local isometry $U=U_A\otimes U_B$ with $$U_A:\mathcal{H}_A\rightarrow\mathcal{H}_{\tilde{A}_0}\otimes\mathcal{H}_{A_0'}\oplus\mathcal{H}_{\tilde{A}_1}\otimes\mathcal{H}_{\hat A_1},$$ $$U_B:\mathcal{H}_B\rightarrow\mathcal{H}_{\tilde{B}_0}\otimes\mathcal{H}_{B_0'}\oplus\mathcal{H}_{\tilde{B}_1}\otimes\mathcal{H}_{\hat B_1}$$
such that
\begin{align}
    &U[E_{xa}\otimes \Id_B\ket{\psi}_{AB}]=(\tilde{E}_{xa}\otimes \Id_{\tilde{B}_0}\ket{\tilde{\psi}}_{\tilde{A}_0\tilde{B}_0}) \ket{\aux_0}_{\hat A_0\hat B_0}\oplus_{\Delta}(\overline{\tilde{E}}_{xa}\otimes \Id_{\tilde{B}_1}\ket{\overline{\tilde{\psi}}}_{\tilde{A}_1\tilde{B}_1}) \ket{\aux_1}_{\hat A_1\hat B_1},\label{eq:combin2A}\\
    &U[\Id_A\otimes F_{yb}\ket{\psi}_{AB}]=(\Id_{\tilde{A}_0}\otimes \tilde{F}_{yb}\ket{\tilde{\psi}}_{\tilde{A}_0\tilde{B}_0}) \ket{\aux_0}_{\hat A_0\hat B_0}\oplus_{\Delta}(\Id_{\tilde{A}_1}\otimes \overline{\tilde{F}}_{yb}\ket{\overline{\tilde{\psi}}}_{\tilde{A}_1\tilde{B}_1}) \ket{\aux_1}_{\hat A_1\hat B_1}
\end{align}
hold for all $a,b,x,y$, where $\ket{\aux_{0,1}}$ are subnormalized state (not necessarily orthogonal): $\braket{\aux_0|\aux_0}+\braket{\aux_1|\aux_1}=1$.
\label{def:complex1}
\end{definition}

The following lemma shows the equivalence between the above definitions. In this paper we will primarily work with Definition \ref{def:complex2}.

\begin{lemma}
Definitions \ref{def:complex2} and \ref{def:complex1} are equivalent.
\end{lemma}

\begin{proof}
We show that Eq. \eqref{eq:combin1A1} and Eq. \eqref{eq:combin2A} implies each other, and the rest can be proved similarly. {In Eq. \eqref{eq:combin2A}, we assume $\cH_{\hat A_0\hat B_0}\cong\cH_{\hat A_1\hat B_1}$ as we can always extend the smaller space to the larger one. Let 
\begin{align*}
    v_0:=&(\tilde{E}_{xa}\otimes \Id_{\tilde{B}_0}\ket{\tilde{\psi}}_{\tilde{A}_0\tilde{B}_0}) \ket{\aux_0}_{\hat A_0\hat B_0}\in \cH_{\tilde{A}_0\tilde{B}_0\hat A_0\hat B_0}=:V_0,\\
    v_1:=&(\overline{\tilde{E}}_{xa}\otimes \Id_{\tilde{B}_1}\ket{\overline{\tilde{\psi}}}_{\tilde{A}_1\tilde{B}_1}) \ket{\aux_1}_{\hat A_1\hat B_1}\in \cH_{\tilde{A}_1\tilde{B}_1\hat A_1\hat B_1}=:V_1,
\end{align*}
then $V_{0}$ and $V_{1}$ has the same dimension, and $V_0\oplus V_1\cong V_0\otimes\C^2$. Then 
$$
{\rm Eq.}~\eqref{eq:combin1A1}=v_0\oplus_\Delta v_1\cong v_0\otimes\ket{0}_{\C^2}+v_1\otimes\ket{1}_{\C^2}, {\rm Eq.}~\eqref{eq:combin2A}=v_0\otimes\ket{00}_{A'B'}+v_1\otimes\ket{11}_{A'B'}.
$$
By embedding $\C^2$ into $\mathcal{H}_{A'B'}$ as its subspace spanned by $\{\ket{00},\ket{11}\}$ we have Eq. \eqref{eq:combin1A1} $\implies$ Eq. \eqref{eq:combin2A}. Similarly, by projecting $\mathcal{H}_{A'B'}$ onto its subspace $\operatorname{span}\{\ket{00},\ket{11}\}\cong\C^2$, we have Eq. \eqref{eq:combin2A} $\implies$ Eq. \eqref{eq:combin1A1}.}
\end{proof}

Now we can define \emph{complex self-testing} in terms of complex local dilations.

\begin{definition}[complex self-testing]
    A strategy $\tilde S=(\ket{\tilde\psi},\{\tilde E_{xa}\},\{\tilde F_{yb}\})$ is \emph{complex self-tested} by a correlation $p(a,b|x,y)$ if it is a complex local dilation of any strategy producing $p(a,b|x,y)$.
    \label{def:c-self-testing}
\end{definition}

Like its standard counterpart \cite{baptista2023mathematicalfoundationselftestinglifting}, one can introduce restrictions for the physical strategy and consider complex self-testing with assumptions. Specifically, we define

\begin{definition}[complex self-testing with assumptions]
    A strategy $\tilde S=(\ket{\tilde\psi},\{\tilde E_{xa}\},\{\tilde F_{yb}\})$ is \emph{complex full-rank self-tested} by a correlation $p(a,b|x,y)$ if it is a complex local dilation of any full-rank strategy producing $p(a,b|x,y)$.

    A strategy $\tilde S=(\ket{\tilde\psi},\{\tilde E_{xa}\},\{\tilde F_{yb}\})$ is \emph{complex PVM self-tested} by a correlation $p(a,b|x,y)$ if it is a complex local dilation of any projective strategy producing $p(a,b|x,y)$.
\end{definition}
\section{Basic properties, and lifting assumptions in complex self-testing}
\label{sec:assumptions}

In this section we establish basic properties of complex local dilations that will be used in the remainder of this paper. Using them we will show that complex self-testing is free from PVM (the physical strategy performs projective measurement) and full-rank (the physical strategy employs a state with full Schmidt rank) assumptions, just like standard self-testing.\footnote{For a detailed discussion about assumptions in self-testing, see \cite[Defnition 2.3, 2.8]{baptista2023mathematicalfoundationselftestinglifting}}

The following two propositions show that a local dilation preserves (exact) support-preservingness and projectiveness.

\begin{proposition}[Counterpart of Proposition 3.4 in \cite{baptista2023mathematicalfoundationselftestinglifting}]
    If $S\xhookrightarrow{}_{\C} \tilde{S}$, then $S$ is support preserving if and only if $\tilde{S}$ is support-preserving.
    \label{prop:cinvariantSupp}
\end{proposition}

\begin{proof}
    We will show support-preservingness via Lemma \ref{lem:carryover}. In the remainder of this proof, we will construct operators $\hat{E}_{xa}$ for each direction, and the construction of $\hat{F}_{yb}$ is analogous.
    
    The `If' direction: Since $\tilde{S}$ is support-preserving, there exist operators $\hat{\tilde E}_{xa}$ such that $\tilde{E}_{xa}\otimes \Id\ket{\tilde\psi}=\Id\otimes\hat{\tilde{E}}_{xa}\ket{\tilde\psi}$ for all $x,a$, and therefore $\overline{\tilde{E}_{xa}}\otimes \Id\ket{\overline{\tilde\psi}}=\Id\otimes\overline{\hat{\tilde{E}}_{xa}}\ket{\overline{\tilde\psi}}$. Construct operators
    $$
    \hat{E}_{xa}:=U_B^*[(\hat{\tilde{E}}_{xa}\otimes \proj{0}_{B'}+\overline{\hat{\tilde{E}}_{xa}}\otimes \proj{1}_{B'})\otimes \Id_{\hat B}]U_B,~\forall~x,a.
    $$
    Then,
    \begin{align*}
        (\Id_A\otimes\hat{E}_{xa})\ket{\psi}=&(U_A^*U_A\otimes\hat{E}_{xa})\ket{\psi}\\
        =&(U_A^*\otimes U_B^*)(\Id_{\tilde{A}}\otimes\hat{\tilde{E}}_{xa})\ket{\tilde\psi}\ket{\aux_0}\ket{00}+(\Id_{\tilde A}\otimes \overline{\hat{\tilde{E}}_{xa}})\ket{\overline{\tilde\psi}}\ket{\aux_1}\ket{11}\\
        =&(U_A^*\otimes U_B^*)(U_A\otimes U_B)(E_{xa}\otimes \Id_B)\ket{\psi}=(E_{xa}\otimes \Id_B)\ket{\psi}.
    \end{align*}
    Hence, $S$ is support-preserving.

    The `Only if' direction: Since $S$ is support-preserving, there exist operators $\hat{E}_{xa}$ such that $E_{xa}\otimes \Id\ket{\psi}=\Id\otimes\hat{E}_{xa}\ket{\psi}$ for all $x,a$. Consider operators
    $$
    \hat{\tilde{E}}_{xa}:=U_B\hat{E}_{xa}U_B^*(\proj{0}_{B'}\otimes \Id_{\tilde{B},\hat B})+\overline{U_B}\overline{\hat{E}_{xa}}U_B^\intercal(\proj{1}_{B'}\otimes \Id_{\tilde{B},\hat B})\in L(\mathcal{H}_{\tilde{B},\hat B,B'}),~\forall~x,a.
    $$
    It holds that
    \begin{align*}
        &(\Id_{\tilde{A},\hat A,A'}\otimes\hat{\tilde{E}}_{xa})\ket{\tilde\psi}\ket{\aux_0}\ket{00}+(\Id_{\tilde{A},A',A'}\otimes\overline{\hat{\tilde{E}}_{xa}})\ket{\overline{\tilde\psi}}\ket{\aux_1}\ket{11}\\
        =&(\Id_{\tilde{A},\hat A,A'}\otimes U_B\hat{{E}}_{xa}U_B^*)(\ket{\tilde\psi}\ket{\aux_0}\ket{00}+\ket{\overline{\tilde\psi}}\ket{\aux_1}\ket{11})\\
        =&(\Id_{\tilde{A},\hat A,A'}\otimes U_B\hat{{E}}_{xa}U_B^*)U\ket{\psi}\\
        =&U(\Id_A\otimes \hat{E}_{xa})\ket{\psi}\\
        =&U(E_{xa}\otimes \Id_B)\ket{\psi}\\
        =&(\tilde{E}_{xa}\otimes \Id_{\tilde B}\ket{\tilde{\psi}})\ket{\aux_0}\ket{00}+(\overline{\tilde{E}}_{xa}\otimes \Id_{\tilde B}\ket{\overline{\tilde{\psi}}})\ket{\aux_1}\ket{11}.
    \end{align*}
    So $(\Id_{\tilde{A},\hat A,A'}\otimes\hat{\tilde{E}}_{xa})\ket{\tilde\psi}\ket{\aux_0}\ket{00}=(\tilde{E}_{xa}\otimes \Id_{\tilde B}\ket{\tilde{\psi}})\ket{\aux_0}\ket{00}$, which means that the operators $\tilde{E}_{xa}\otimes \Id_{\hat A,A'}$ and $\Pi_{\tilde A}\otimes \Pi_{\aux_0,A}\otimes \proj{0}_{A'}$ commute. Notice that 
    $$
    [\tilde{E}_{xa}\otimes \Id_{\hat A,A'},\Pi_{\tilde A}\otimes \Pi_{\aux_0,A}\otimes\proj{0}_{A'}]=[\tilde{E}_{xa},\Pi_{\tilde A}]\otimes\Pi_{\aux,A}\otimes\proj{0}_{A'}.
    $$
    Hence, $\tilde{E}_{xa}$ and $\Pi_{\tilde A}$ commute, so $\tilde{S}$ is support-preserving.
\end{proof}

\begin{proposition}[Counterpart of Proposition 3.6 in \cite{baptista2023mathematicalfoundationselftestinglifting}]
\label{prop:cinvariantProj}
    If $S\xhookrightarrow{}_{\C} \tilde{S}$, then $S$ is 0-projective if and only if $\tilde{S}$ is 0-projective.
\end{proposition}

\begin{proof}
    Note that 
    \begin{align*}
        U[E_{xa}\otimes \Id_B\ket{\psi}_{AB}]=&(\tilde{E}_{xa}\otimes \Id_{\tilde B}\ket{\tilde{\psi}}_{\tilde{A}\tilde{B}})\ket{\aux_0}_{\hat{A}\hat{B}}\ket{00}_{A'B'}\\
    +&(\overline{\tilde{E}_{xa}}\otimes \Id_{\tilde B}\ket{\overline{\tilde{\psi}}}_{\tilde{A}\tilde{B}})\ket{\aux_1}_{\hat{A}\hat{B}}\ket{11}_{A'B'},\\
    U[(\Id_A-E_{xa})\otimes \Id_B\ket{\psi}_{AB}]=&((\Id_{\tilde A}-\tilde{E}_{xa})\otimes \Id_{\tilde B}\ket{\tilde{\psi}}_{\tilde{A}\tilde{B}})\ket{\aux_0}_{\hat{A}\hat{B}}\ket{00}_{A'B'}\\
    +&((\Id_{\tilde A}-\overline{\tilde{E}_{xa}}\otimes \Id_{\tilde B}\ket{\overline{\tilde{\psi}}}_{\tilde{A}\tilde{B}})\ket{\aux_1}_{\hat{A}\hat{B}}\ket{11}_{A'B'}.
    \end{align*}
    Taking the inner product of the above two equations results in
    \begin{align*}
        \braket{\psi|E_{xa}(\Id_A-E_{xa})\otimes\Id_B|\psi}=&\braket{\tilde\psi|\tilde E_{xa}(\Id_{\tilde A}-\tilde E_{xa})\otimes\Id_{\tilde B}|\tilde\psi}\braket{\aux_0|\aux_0}\\
        +&\braket{\tilde\psi|\overline{\tilde E_{xa}}(\Id_{\tilde A}-\overline{\tilde E_{xa}})\otimes\Id_{\tilde B}|\overline{\tilde\psi}}\braket{\aux_1|\aux_1}.
    \end{align*}
    On the one hand, if $\braket{\tilde\psi|\tilde E_{xa}(\Id_{\tilde A}-\tilde E_{xa})\otimes\Id_{\tilde B}|\tilde\psi}=0$, then likewise $\braket{\overline{\tilde\psi}|\overline{\tilde E_{xa}}(\Id_{\tilde A}-\overline{\tilde E_{xa}})\otimes\Id_{\tilde B}|\overline{\tilde\psi}}=0$, and hence $\braket{\psi|E_{xa}(\Id_A-E_{xa})\otimes\Id_B|\psi}=0$. On the other hand, since both $\tilde E_{xa}(\Id_{\tilde A}-\tilde E_{xa})$ and $(\Id_{\tilde A}-\overline{\tilde E_{xa}})$ are positive semidefinite, $\braket{\psi|E_{xa}(\Id_A-E_{xa})\otimes\Id_B|\psi}=0$ implies  \\$\braket{\tilde\psi|\tilde E_{xa}(\Id_{\tilde A}-\tilde E_{xa})\otimes\Id_{\tilde B}|\tilde\psi}=0$, hence $\tilde S$ is 0-projective. Thus, we conclude that $S$ is 0-projective if and only if $\tilde S$ is 0-projective.
\end{proof}

With Propositions \ref{prop:cinvariantProj} and \ref{prop:cinvariantSupp} at hand, we are ready to present the ``lifting assumption'' theorems.

\begin{theorem}[Counterpart of Theorem 4.3 in \cite{baptista2023mathematicalfoundationselftestinglifting}]
    Let $p(a,b|x,y)$ be a quantum correlation. If $p(a,b|x,y)$ pure PVM complex self-tests a full-rank canonical strategy $\tilde S$, then $\tilde S$ is projective, and $p(a,b|x,y)$ pure complex self-tests $\tilde S$.
\end{theorem}

\begin{proof}
    For any $S$ (with a pure state) that generates $p(a,b|x,y)$, consider its Naimark dilation $S_{\operatorname{Naimark}}$. Since $p(a,b|x,y)$ pure PVM self-tests $\tilde S$, it holds that $S_{\operatorname{Naimark}}\xhookrightarrow{}_{\C}\tilde S$. By Proposition \ref{prop:cinvariantProj}, $\tilde S$ is 0-projective, thus projective (since it is full-rank).

    Again for any $S$ (with a pure state) that generates $p(a,b|x,y)$, its Naimark dilation satisfies $S_{\operatorname{Naimark}}\xhookrightarrow{}_{\C}\tilde S$. Note that $\tilde S$ is assumed to be full-rank (thus support-preserving); then, $S_{\operatorname{Naimark}}$ is support-preserving by Proposition \ref{prop:cinvariantSupp}. Consequently, $S$ is support-preserving by \cite[Theorem 3.18]{baptista2023mathematicalfoundationselftestinglifting}. Then, $S\xhookrightarrow{}S_{\operatorname{Naimark}}$ by \cite[Proposition 3.17]{baptista2023mathematicalfoundationselftestinglifting}. 
    By transitivity, $S\xhookrightarrow{}_{\C}\tilde{S}$. So we conclude that $p(a,b|x,y)$ also pure self-tests $\tilde S$.
\end{proof}

\begin{theorem}[Counterpart of Theorem 4.5 in \cite{baptista2023mathematicalfoundationselftestinglifting}]
    Let $p(a,b|x,y)$ be a quantum correlation. If $p(a,b|x,y)$ pure full-rank complex self-tests a PVM canonical strategy $\tilde S$, then $\tilde S$ is support-preserving, and $p(a,b|x,y)$ pure complex self-tests $\tilde S$.
\end{theorem}

\begin{proof}
    For any $S$ (with a pure state) that generates $p(a,b|x,y)$, consider its restriction to the support $S_{\operatorname{res}}$. Since $p(a,b|x,y)$ pure full-rank self-tests $\tilde S$, it holds that $S_{\operatorname{res}}\xhookrightarrow{}_{\C}\tilde S$. By Proposition \ref{prop:cinvariantSupp}, $\tilde S$ is support-preserving.

    Again for any $S$ (with a pure state) that generates $p(a,b|x,y)$, its restriction satisfies $S_{\operatorname{Naimark}}\xhookrightarrow{}_{\C}\tilde S$. Note that $\tilde S$ is assumed to be projective (thus 0-projective); then, $S_{\operatorname{res}}$ is 0-projective by Proposition \ref{prop:cinvariantProj}. Consequently, $S$ is 0-projective by \cite[Theorem 3.9]{baptista2023mathematicalfoundationselftestinglifting}. Then, $S\xhookrightarrow{}S_{\operatorname{res}}$ by \cite[Proposition 3.8]{baptista2023mathematicalfoundationselftestinglifting}. 
    By transitivity, $S\xhookrightarrow{}_{\C}\tilde{S}$. So we conclude that $p(a,b|x,y)$ also pure self-tests $\tilde S$.
\end{proof}

We remark that this paper considers strategies with a pure state, therefore lifting the purity assumption in complex self-testing is beyond its scope, which we leave for future work.

Finally, we discuss the real simulation of a quantum strategy \cite{McKague_2009}, a special type of complex local dilation that may be of independent interest. It shows that any quantum correlation admits a real quantum realization (all matrix entries are real) via complex local dilation. Let $\ket{\pm i}:=(\ket{0}\pm i\ket{1})/\sqrt{2}$ be the eigenstate of Pauli matrix $\sigma_Y$. 
\begin{definition}[real simulation]\label{def:realsim}
    Let $S=(\ket{\psi},\{ E_{xa}\},\{ F_{yb}\})$ be a complex strategy. The real simulation $S_R$ of $S$ is defined as $S_R:=(\ket{\psi_R},\{E_{R,xa}\},\{F_{R,yb}\})$, where
\begin{align*}
    &\ket{\psi_R}:=(\ket{\psi}\ket{+i+i}+\ket{\overline{\psi}}\ket{-i-i})/\sqrt{2},\\
    &E_{R,xa}:={E}_{xa}\otimes\ket{+i}\!\bra{+i}+\overline{{E}_{xa}}\otimes\ket{-i}\!\bra{-i},\\
    &F_{R,yb}:={F}_{yb}\otimes\ket{+i}\!\bra{+i}+\overline{{F}_{yb}}\otimes \ket{-i}\!\bra{-i}.
\end{align*}
\end{definition}
It is straightforward to verify that $\ket{\psi_R},A_{R,xa},B_{R,yb}$ all have real entries, and $S_R$ gives the same correlation as $S$. By the definition of complex dilation \ref{def:c-self-testing}. it holds that $S_R\xhookrightarrow{}_{\C}S$. We remark that the auxiliary state $\ket{\pm i}$ is not strictly necessary; any state $\ket{\phi}$ satisfying $\braket{\phi|\overline{\phi}}=0$ would suffice. 

The following property relates real simulations of two strategies when one is a complex local dilations of another.

\begin{proposition}
    If $S\xhookrightarrow{}_{\C}\tilde{S}$, then $S_R\xhookrightarrow{}\tilde{S}_R$.
    \label{prop:c-realsim}
\end{proposition}
\begin{proof}
    Given that $S\xhookrightarrow{}_{\C}\tilde{S}$, there exist local isometries $V_A,V_B$ and auxiliary states $\ket{\aux_0},\ket{\aux_1}$ that satisfy the complex local dilation relations. Now consider the action of $V_{A,R}:=\proj{+i}\otimes V_A+\proj{-i}\otimes\overline{V_A},V_{B,R}:=V_B\otimes\proj{+i}+\overline{V_B}\otimes\proj{-i}$ on $S_R$. We have
    \begin{align*}
        &\left(V_{A,R}\otimes V_{B,R}\right)\left(E_{xa,R}\otimes F_{yb,R}\ket{\psi_R}\right)\\
        =&\left(V_{A,R}\otimes V_{B,R}\right)\frac{1}{\sqrt{2}}\left(\ket{+i+i}(E_{xa}\otimes F_{yb}\ket{\psi})+\ket{-i-i}(\overline{E_{xa}}\otimes \overline{F_{yb}}\ket{\overline\psi})\right)\\
        =&\frac{1}{\sqrt{2}}\left(\ket{+i+i}(V_A\otimes V_B)(E_{xa}\otimes F_{yb}\ket{\psi})+\ket{-i-i}(\overline{V_A}\otimes \overline{V_B})(\overline{E_{xa}}\otimes \overline{F_{yb}}\ket{\overline\psi})\right)\\
        =&\frac{1}{\sqrt{2}}(\ket{+i+i}\ket{00}\ket{\aux_0}(\tilde{E}_{xa}\otimes\tilde{F}_{yb}\ket{\tilde\psi})+\ket{+i+i}\ket{11}\ket{\aux_1}(\overline{\tilde{E}_{xa}}\otimes\overline{\tilde{F}_{yb}}\ket{\overline{\tilde\psi}})+\\
        &\ket{-i-i}\ket{00}\ket{\aux_0}(\overline{\tilde{E}_{xa}}\otimes\overline{\tilde{F}_{yb}}\ket{\overline{\tilde\psi}})+\ket{-i-i}\ket{11}\ket{\aux_1}({\tilde{E}_{xa}}\otimes{\tilde{F}_{yb}}\ket{\tilde\psi}))\\
        =&\frac{1}{\sqrt{2}}\left(\ket{00}\ket{\aux_0}\right)(\ket{+i+i}\tilde{E}_{xa}\otimes\tilde{F}_{yb}\ket{\tilde\psi}+\ket{-i-i}\overline{\tilde{E}_{xa}}\otimes\overline{\tilde{F}_{yb}}\ket{\overline{\tilde\psi}})+\\
        &\frac{1}{\sqrt{2}}\left(\ket{11}\ket{\aux_1}\right)(\ket{-i-i}\tilde{E}_{xa}\otimes\tilde{F}_{yb}\ket{\tilde\psi}+\ket{+i+i}\overline{\tilde{E}_{xa}}\otimes\overline{\tilde{F}_{yb}}\ket{\overline{\tilde\psi}}).
    \end{align*}
    Let $U_i$ be the 2-dimensional unitary that maps $\ket{\pm i}$ to $\ket{\mp i}$. Then consider the action of the local unitary $U:=\proj{0}\otimes \Id_{\aux}\otimes \Id_i+\proj{1}\otimes \Id_{\aux}\otimes U_i$. Clearly, $U\otimes U$ keeps $\ket{00}\ket{\aux_0}\ket{\pm i\pm i}$ unchanged, and maps $\ket{11}\ket{\aux_1}\ket{\pm i\pm i}$ to $\ket{11}\ket{\aux_1}\ket{\mp i\mp i}$. Therefore,
    \begin{align*}
        &(U\otimes \Id_{\tilde{A}}\otimes U\otimes \Id_{\tilde B})\left(V_{A,R}\otimes V_{B,R}\right)\left(E_{xa,R}\otimes F_{yb,R}\ket{\psi_R}\right)\\
        =&\frac{1}{\sqrt{2}}\left(\ket{00}\ket{\aux_0}\right)(\ket{+i+i}\tilde{E}_{xa}\otimes\tilde{F}_{yb}\ket{\tilde\psi}+\ket{-i-i}\overline{\tilde{E}_{xa}}\otimes\overline{\tilde{F}_{yb}}\ket{\overline{\tilde\psi}})+\\
        &\frac{1}{\sqrt{2}}\left(\ket{11}\ket{\aux_1}\right)(\ket{+i+i}\tilde{E}_{xa}\otimes\tilde{F}_{yb}\ket{\tilde\psi}+\ket{-i-i}\overline{\tilde{E}_{xa}}\otimes\overline{\tilde{F}_{yb}}\ket{\overline{\tilde\psi}})\\
        =&(\ket{00}\ket{\aux_0}+\ket{11}\ket{\aux_1})(\tilde{E}_{xa,R}\otimes\tilde{F}_{yb,R}\ket{\tilde\psi_R}).
    \end{align*}
    Also notice that $(\ket{00}\ket{\aux_0}+\ket{11}\ket{\aux_1})$ is a unit vector. We conclude that $S_R\xhookrightarrow{}\tilde{S}_R$ via local isometry $(U\otimes \Id_{\tilde{A}}\otimes U\otimes \Id_{\tilde B})\left(V_{A,R}\otimes V_{B,R}\right)$.
\end{proof}

Finally, we point out that $S\xhookrightarrow{}S_R$ fails in general for a full-rank $S$. This can be seen from the fact that strategies connected by a local dilation have the same higher moments (\cite[Proposition 4.8]{paddock2023operatoralgebraic}), which are necessarily real-valued for $S_R$. This is possible only if $S$ has real-valued moments.
\section{An operator-algebraic characterization}
\label{sec:c-operator-algebraic}

Here we present the operator-algebraic picture of complex self-testing. We first show that a complex self-test is equivalent to the reference strategy having a unique real part of higher moments. Then, we describe this in terms of real states on real C* algebras. 

\subsection{Complex self-testing implies a unique real part of moments}

We start with the simpler direction of the aformentioned equivalence.

\begin{proposition}
\label{prop:hom}
    If a support-preserving strategy $\tilde S$ is complex self-tested by $p(a,b|x,y)$, then for any $k,\ell\in\mathbb{Z}^+$, for any $\vec x,\vec a$ of length $k$, and any $\vec y,\vec b$ of length $\ell$,
    $$
    \re\braket{\psi|E_{\vec{x}\vec{a}}\otimes F_{\vec y \vec b}|\psi}
    $$
    is the same across all strategies $S$ producing $p(a,b|x,y)$.
\end{proposition}

\begin{proof}
    Let $U=U_A\otimes U_B$ be the isometry and $\ket{\aux_{0,1}}$ be the auxiliary state from the complex self-test. We will prove that \begin{align}U(E_{\vec{x}\vec{a}}\otimes\Id_B)\ket{\psi}=(\tilde{E}_{\vec{x}\vec{a}}\otimes \Id_{\tilde B}\ket{\tilde{\psi}})\ket{\aux_0}\ket{00}+(\overline{\tilde{E}_{\vec{x}\vec{a}}}\otimes \Id_{\tilde B}\ket{\overline{\tilde{\psi}}})\ket{\aux_1}\ket{11}
    \label{eq:1}
    \end{align}
    holds for all $\vec{x},\vec{a}$ by induction. First, from complex self-testing, for any single $x,a$ it holds that 
    \begin{align}
        &(\tilde{E}_{xa}\otimes \Id_{\tilde B}\ket{\tilde{\psi}}_{\tilde{A}\tilde{B}})\ket{\aux_0}_{\hat{A}\hat{B}}\ket{00}_{A'B'}+(\overline{\tilde{E}_{xa}}\otimes \Id_{\tilde B}\ket{\tilde{\psi}}_{\tilde{A}\tilde{B}})\ket{\aux_1}_{\hat{A}\hat{B}}\ket{11}_{A'B'}\nonumber\\
        =&(U_A\otimes U_B)(E_{xa}\otimes \Id)\ket{\psi}\nonumber\\
        =&(U_AE_{xa}U_A^*\otimes \Id)(U_A\otimes U_B)\ket{\psi}\nonumber\\
        =&(U_AE_{xa}U_A^*\otimes \Id)\ket{\tilde{\psi}}_{\tilde{A}\tilde{B}}\ket{\aux_0}_{\hat{A}\hat{B}}\ket{00}_{A'B'}+(U_AE_{xa}U_A^*\otimes \Id)\ket{\tilde{\psi}}_{\tilde{A}\tilde{B}}\ket{\aux_1}_{\hat{A}\hat{B}}\ket{11}_{A'B'}.\label{eq:3}
    \end{align}
    It is clear that Eq. \eqref{eq:1} is true with words of length $0$ or $1$. Suppose it is true for words of length $k$. Then for any $\vec x,\vec a$ of length $k$, we have that 
    \begin{align*} &U(E_{x_{k+1}a_{k+1}}E_{\vec{x}\vec{a}}\otimes\Id_B)\ket{\psi}\\
        =&(U_AE_{x_{k+1}a_{k+1}}U_A^*\otimes \Id)U(E_{\vec{x}\vec{a}}\otimes\Id_B)\ket{\psi}\\
        =&(U_AE_{x_{k+1}a_{k+1}}U_A^*\otimes \Id)((\tilde{E}_{\vec{x}\vec{a}}\otimes \Id_{\tilde B}\ket{\tilde{\psi}})\ket{\aux_0}\ket{00}+(\overline{\tilde{E}_{\vec{x}\vec{a}}}\otimes \Id_{\tilde B}\ket{\overline{\tilde{\psi}}})\ket{\aux_1}\ket{11})\\
        =&(U_AE_{x_{k+1}a_{k+1}}U_A^*\otimes \Id)((\Id_{\tilde A}\otimes \hat{\tilde{E}}_{\vec{x}\vec{a}}\ket{\tilde{\psi}})\ket{\aux_0}\ket{00}+(\Id_{\tilde A}\otimes \hat{\overline{\tilde{E}_{\vec{x}\vec{a}}}} ~\ket{\overline{\tilde{\psi}}})\ket{\aux_1}\ket{11})\\
        =&(\Id_{\tilde A\hat{A}A'}\otimes \hat{\tilde{E}}_{\vec{x}\vec{a}}\otimes\Id_{\hat{B}B'})(U_AE_{x_{k+1}a_{k+1}}U_A^*\otimes \Id)(\ket{\tilde{\psi}}\ket{\aux_0}\ket{00})\\
        +&(\Id_{\tilde A\hat{A}A'}\otimes \hat{\overline{\tilde{E}_{\vec{x}\vec{a}}}}~\otimes\Id_{\hat{B}B'})(U_AE_{x_{k+1}a_{k+1}}U_A^*\otimes \Id)(\ket{\overline{\tilde{\psi}}}\ket{\aux_1}\ket{11}).
    \end{align*}
    The third equation uses the fact that $\tilde S$ is support preserving and Lemma \ref{lem:carryover}. Multiplying both sides of Eq. \eqref{eq:3} by $(\hat{\tilde{E}}_{\vec{x}\vec{a}}\otimes\proj{0}_{B'}+\hat{\overline{\tilde{E}_{\vec{x}\vec{a}}}}~\otimes\proj{1}_{B'})\otimes\Id_{\tilde A\hat{A}A'\hat{B}}$, we get
    \begin{align*}
        &(\Id_{\tilde A\hat{A}A'}\otimes\hat{\tilde{E}}_{\vec{x}\vec{a}}\otimes\Id_{\hat{B}B'})(U_AE_{x_{k+1}a_{k+1}}U_A^*\otimes \Id)(\ket{\tilde{\psi}}\ket{\aux_0}\ket{00})\\
        +&(\Id_{\tilde A\hat{A}A'}\otimes \hat{\overline{\tilde{E}_{\vec{x}\vec{a}}}}~\otimes\Id_{\hat{B}B'})(U_AE_{x_{k+1}a_{k+1}}U_A^*\otimes \Id)(\ket{\overline{\tilde{\psi}}}\ket{\aux_1}\ket{11})\\
        =&(\tilde{E}_{x_{k+1}a_{k+1}}\otimes \hat{\tilde{E}}_{\vec{x}\vec{a}}\ket{\tilde{\psi}})\ket{\aux_0}\ket{00}+(\overline{\tilde{E}_{x_{k+1}a_{k+1}}}\otimes \hat{\overline{\tilde{E}_{\vec{x}\vec{a}}}}~\ket{\overline{\tilde{\psi}}})\ket{\aux_1}\ket{11}\\
        =&(\tilde{E}_{x_{k+1}a_{k+1}}{\tilde{E}}_{\vec{x}\vec{a}}\otimes \Id_{\tilde B}\ket{\tilde{\psi}})\ket{\aux_0}\ket{00}+(\overline{\tilde{E}_{x_{k+1}a_{k+1}}}{\overline{\tilde{E}_{\vec{x}\vec{a}}}}\otimes \Id_{\tilde B}\ket{\overline{\tilde{\psi}}})\ket{\aux_1}\ket{11}
    \end{align*}
    Hence, Eq. \eqref{eq:1} holds for all words. We can prove similar statement for Bob's operator. Then, 
    \begin{align*}
    U(E_{\vec x \vec a}\otimes F_{\vec y \vec b})\ket{\psi}=&(U_AE_{\vec x \vec a}U_A^*\otimes\Id_{\tilde B\hat{B}B'})U(\Id\otimes F_{\vec y \vec b})\ket{\psi}\\
    =&(U_AE_{\vec x \vec a}U_A^*\otimes\Id_{\tilde B\hat{B}B'})[(\Id_{\tilde A}\otimes \tilde{F}_{\vec y \vec b})\ket{\tilde{\psi}}\ket{\aux_0}\ket{00}+(\Id_{\tilde A}\otimes \overline{\tilde{F}_{\vec y \vec b}})\ket{\overline{\tilde{\psi}}}\ket{\aux_1}\ket{11}]\\
    =&(\Id_{\tilde A\hat{A}A'}\otimes \tilde{F}_{\vec y \vec b})(U_AE_{\vec x \vec a}U_A^*\otimes\Id_{\tilde B\hat{B}B'})\ket{\tilde{\psi}}\ket{\aux_0}\ket{00}\\
    +&(\Id_{\tilde A\hat{A}A'}\otimes \overline{\tilde{F}_{\vec y \vec b}})(U_AE_{\vec x \vec a}U_A^*\otimes\Id_{\tilde B\hat{B}B'})\ket{\overline{\tilde{\psi}}}\ket{\aux_1}\ket{11}\\
    =&(\tilde{E}_{\vec x \vec a}\otimes \tilde{F}_{\vec y \vec b})\ket{\tilde{\psi}}\ket{\aux_0}\ket{00}+(\overline{\tilde{E}_{\vec x \vec a}}\otimes \overline{\tilde{F}_{\vec y \vec b}})\ket{\overline{\tilde{\psi}}}\ket{\aux_1}\ket{11}.
    \end{align*}
    Note that $U\ket{\psi}=\ket{\tilde{\psi}}\ket{\aux_0}\ket{00}+\ket{\overline{\tilde{\psi}}}\ket{\aux_1}\ket{11}$. Taking the inner product of the two sides respectively, we get
    \begin{align}
    \braket{\psi|E_{\vec x \vec a}\otimes F_{\vec y \vec b}|\psi}=\braket{\tilde\psi|\tilde E_{\vec x \vec a}\otimes \tilde F_{\vec y \vec b}|\tilde\psi}|\ket{\aux_0}|^2+\braket{\overline{\tilde\psi}|\overline{\tilde E_{\vec x \vec a}}\otimes \overline{\tilde F_{\vec y \vec b}}|\overline{\tilde\psi}}|\ket{\aux_1}|^2.\label{eq:4}
    \end{align}
    The final statement is achieved by taking the real part of both sides of the equation.
\end{proof}

We remark that Eq. \eqref{eq:4} also implies that if the canonical strategy $\tilde S$ has all higher moments real, then so does $S$, in which case complex self-testing reduces to standard self-testing.

\subsection{Unique real part of moments implies complex self-testing}

Next, we prove the converse of Proposition \ref{prop:hom}, given an extreme $p(a,b|x,y)$.

\begin{proposition}
    Suppose a correlation $p(a,b|x,y)$ is extreme in $C_q$. If all strategies producing $p(a,b|x,y)$ have the same real parts of their moments, then there is a canonical $\tilde S$ such that $\tilde S$ is complex self-tested by correlation $p(a,b|x,y)$.
    \label{prop:c-realstate2}
\end{proposition}

The proof of Proposition \ref{prop:c-realstate2} relies on 
the following two lemmas about real polynomials on $E_{xa}\otimes F_{yb}$ 
and irreducible strategies.

\begin{lemma}
    Let $S,S'$ be irreducible strategies with the same real parts of their moments, and let $f$ be any real polynomial. Then $f(E_{xa}\otimes F_{yb})=0$ if and only if $f(E'_{xa}\otimes F'_{yb})=0$.

    Consequently, if $g(E_{xa}\otimes F_{yb})=i\Id$ for some real polynomial $g$, then $g(E'_{xa}\otimes F'_{yb})=\pm i\Id$.
    \label{lem:realpoly}
\end{lemma}
\begin{proof}
    Due to the symmetry, it suffice to show that $f(E_{xa}\otimes F_{yb})=0$ implies $ f(E'_{xa}\otimes F'_{yb})=0$. 
    For any real self-adjoint polynomial $g$, 
    \begin{equation}
\braket{\psi|g(E_{xa}\otimes F_{yb})|\psi}
=\re\braket{\psi|g(E_{xa}\otimes F_{yb})|\psi}
=\re\braket{\psi'|g(E_{xa}'\otimes F_{yb}')|\psi'}
=\braket{\psi'|g(E_{xa}'\otimes F_{yb}')|\psi'}
\label{eq:realparts}
    \end{equation}
    by the assumption and self-adjointness.
Let $h$ be a real polynomial. Since $h^*f^*fh$ is a real self-adjoint polynomial, Eq.~\eqref{eq:realparts} and positive semidefiniteness of $(h^*f^*fh)(E'_{xa}\otimes F'_{yb})$ imply
\begin{alignat*}{2}
&&\braket{\psi|(h^*f^*fh)(E_{xa}\otimes F_{yb})|\psi}&=0\\
\implies&&
\braket{\psi'|(h^*f^*fh)(E'_{xa}\otimes F'_{yb})|\psi'}&=0\\
\implies&&
(fh)(E'_{xa}\otimes F'_{yb})\ket{\psi'}&=0.
\end{alignat*}
Therefore, $f(E'_{xa}\otimes F'_{yb})\cdot h(E'_{xa}\otimes F'_{yb})\ket{\psi'}=0$ holds for every real polynomial $h$. Since $S'$ is irreducible, the set $\{h(E'_{xa}\otimes F'_{yb})\ket{\psi'}: h\text{ a real polynomial} \}$ spans $\cH_{A'}\otimes \cH_{B'}$. Therefore, $f(E'_{xa}\otimes F'_{yb})=0$. 

To prove the second part of the lemma, assume $g(E_{xa}\otimes F_{yb})=i\Id$. Then,
$$(g^2+1)(E_{xa}\otimes F_{yb})=0,\qquad (gh-hg)(E_{xa}\otimes F_{yb})=0\quad \text{for all real polynomials }h.$$
Note that $g^2+1$ and $gh-hg$ are real polynomials. 
By the first part of the lemma, 
$$(g^2+1)(E_{xa}'\otimes F_{yb}')=0,\qquad (gh-hg)(E_{xa}'\otimes F_{yb}')=0\quad \text{for all real polynomials }h.$$
Since $S'$ is irreducible, only scalar operators commute with all $E_{xa}'\otimes F_{yb}'$, hence $g(E_{xa}'\otimes F_{yb}')=\alpha \Id$ for some $\alpha\in\C$. 
Finally, $(g^2+1)(E_{xa}'\otimes F_{yb}')=0$ implies $\alpha=\pm i$.
\end{proof}

\begin{lemma}
If two \emph{irreducible} strategies $S,S'$ have the same real parts of their moments, then they $S'$ is unitarily equivalent to $S$ or $\overline{S}$.
\label{lem:alternative}
\end{lemma}

\begin{proof}
Let $\cA\subset L(\cH_A)\otimes L(H_B)$ and $\cA'\subset L(\cH_{A'})\otimes L(H_{B'})$ be unital real subalgebras generated by $E_{xa}\otimes F_{yb}$ and 
$E_{xa}'\otimes F_{yb}'$, respectively. 
By irreducibility,
\begin{equation}
L(\cH_A)\otimes L(H_B)=
\cA+i\cA,\qquad 
L(\cH_{A'})\otimes L(H_{B'})=\cA'+i\cA',
\label{eq:ext}
\end{equation}
though these sums may not be direct sums. 
By Lemma \ref{lem:realpoly}, the map
$$\phi:\cA\to\cA',\qquad f(E_{xa}\otimes F_{yb})\mapsto f(E_{xa}'\otimes F_{yb}')$$
for real polynomials $f$ is well-defined, and is an isomorphism of real $*$-algebras. We distinguish two cases.
\begin{itemize}
\item Suppose that $f(E_{xa}\otimes F_{yb})\neq i\Id$ for all real polynomials $f$. Then, both sums in \eqref{eq:ext} are direct sums by Lemma \ref{lem:realpoly}. Hence, we can extend $\phi$ to an isomorphism of complex $*$-algebras $\phi:L(\cH_A)\otimes L(H_B)\to L(\cH_{A'})\otimes L(H_{B'})$ via $\phi(X+iY)=\phi(X)+i\phi(Y)$ for $X,Y\in \cA$.
\item Suppose that $f(E_{xa}\otimes F_{yb})= i\Id$ for some real polynomial $f$. Then, $f(E_{xa}'\otimes F_{yb}')= \pm i\Id$ by Lemma \ref{lem:realpoly}. Consequently, $L(\cH_A)\otimes L(H_B)=\cA$ and $L(\cH_{A'})\otimes L(H_{B'})=\cA'$. Note that $\phi:L(\cH_A)\otimes L(H_B)$ is an isomorphism of real $*$-algebras. If $f(E_{xa}'\otimes F_{yb}')= i\Id$, then $\phi(i\Id)=i\Id$, so $\phi$ is an isomorphism of complex $*$-algebras; if $f(E_{xa}'\otimes F_{yb}')= -i\Id$, then $\phi(i\Id)=-i\Id$, so $\overline{\phi}$ is an isomorphism of complex $*$-algebras.
\end{itemize}
In both cases (after replacing $S$ with $\overline{S}$ in the second case if needed), there is an isomorphism of complex $*$-algebras $\phi:L(\cH_A)\otimes L(H_B)\to L(\cH_{A'})\otimes L(H_{B'})$ that maps $E_{xa}\otimes F_{yb}$ to $E_{xa}'\otimes F_{yb}'$. Note that $\phi$ maps $L(\cH_A)\otimes\Id$ onto
$L(\cH_{A'})\otimes\Id$ and 
$\Id\otimes L(\cH_B)$ onto $\Id\otimes L(\cH_{B'})$. Hence, by the Skolem-Noether theorem \cite[Theorem 4.46]{brevsar2014introduction}, there are unitaries $U:\cH_A\to \cH_{A'}$ and $V:\cH_B\to \cH_{B'}$ such that $\phi(Z)=(U\otimes V)Z(U^*\otimes V^*)$ for $Z\in L(\cH_A\otimes \cH_B)$. Finally, since 
$$\braket{\psi'|(U\otimes V)Z(U^*\otimes V^*)|\psi'}
=\braket{\psi'|\phi(Z)|\psi'}
=\braket{\psi|Z|\psi}\quad \text{for all }Z\in L(\cH_A\otimes \cH_B),$$
it follows that $\ket{\psi'}=(U\otimes V)\ket{\psi}$.
\end{proof}

\begin{proof}[Proof of Proposition \ref{prop:c-realstate2}]
Given a quantum strategy $S=(\ket\psi,\{E_{xa}\},\{F_{yb}\})$ producing $p(a,b|x,y)$, decompose it by the fundamental structure theorem of finite-dimensional C*-algebras:
    \begin{align*}
        &E_{xa}=\bigoplus_iE^{(i)}_{xa}\otimes\Id\in\bigoplus_i\cH_A^{(i)}\otimes\cK_A^{(i)},\\
        &F_{yb}=\bigoplus_jF^{(j)}_{yb}\otimes\Id\in\bigoplus_j\cH_B^{(j)}\otimes\cK_B^{(j)},\\
        &\ket{\psi}=\bigoplus_{i,j}\left(\sum_{k,\ell}\sqrt{\lambda^{(ijk\ell)}}\ket{\psi^{(ijk\ell)}}\otimes\ket{\alpha^{(ik)},\beta^{(j\ell)}}\right)\in\bigoplus_{i,j}\cH_A^{(i)}\otimes\cH_B^{(j)}\otimes\cK_A^{(i)}\otimes\cK_B^{(j)},
    \end{align*}
    where positive coefficients $\lambda^{(ijk\ell)}>0$ satisfies $\sum_{i,j,k,\ell}\lambda^{(ijk\ell)}=1$, and $\{\ket{\alpha^{(ik)}}\}_k,\{\ket{\beta^{(j\ell)}}\}_
    \ell$ are orthonormal bases for $\cK_A^{(i)},\cK_B^{(j)}$, respectively. Consider irreducible $S^{(ijk\ell)}:=(\ket{\psi^{ijk\ell}},\{E^{(i)}_{xa}\},\{F^{(j)}_{yb}\})$ and their correlations $p^{(ijk\ell)}$. Then, $\sum_{i,j,k,\ell}\lambda^{(ijk\ell)}p^{(ijk\ell)}=p$. Since $p$ is extreme in $C_q$, we have that $p^{(ijk\ell)}=p$. By our hypothesis, $S^{(ijk\ell)}$ have the same real parts of their moments. Thus, there exists an irreducible $\tilde S=(\ket{\tilde\psi},\{\tilde{E}_{xa}\},\{\tilde{F}_{yb}\})$ that produces $p$.

    Since $\tilde S$ and $S^{(ijk\ell)}$ agree on the real parts of their moments, $S^{(ijk\ell)}$ is unitarily equivalent to $\tilde{S}$ or $\overline{\tilde{S}}$ by Lemma \ref{lem:alternative}.
    As this holds for every quadruple $(ijkl)$, we conclude that $S\xhookrightarrow{}_{\mathbb C}\tilde S$ by Definition \ref{def:complex1}.
\end{proof}

\subsection{The operator-algebraic formulation}

In the language of operator algebras, having a unique real part of moments can be translated to uniqueness of a finite dimensional real state on a certain universal real C* algebra. A real (unital) C* algebra is a Banach *-algebra over $\R$ satisfying the C*-identity $\|a^*a\|=\|a\|^2$, as well as the additional property that $1+a^*a$ is invertible for every $a$ (this ensures a real version of the GNS construction). The framework of real C* algebras shares many similarities with that of complex C* algebras, for instance the GNS construction. For a comprehensive introduction of real C* algebras, see \eg, \cite{goodearl,li2003real}. A \emph{universal} real C* algebra $\cA_{\R}(G,R)$ is a real C* algebra such that (1) the elements of $G$ generate $\cA_{\R}(G,R)$ and satisfy the relations in $R$, and (2) has the following universal property: for any real C*-algebra $\mathcal{B}$ and any set of elements in $\mathcal{B}$ that satisfy the same relations $R$, there exists a unique *-homomorphism from $\mathcal{A}_{\R}(G, R)$ to $\mathcal{B}$ that maps the generators accordingly.

\begin{lemma}
    \label{lem:c*=moments}
    The following statements are equivalent:
    \begin{enumerate}
        \item For any $k,\ell\in\mathbb{N}^+$, for any $\vec x,\vec a$ of length $k$ and $\vec y,\vec b$ of length $\ell$, the real parts of moments
    $$
    \re\braket{\psi|E_{\vec{x}\vec{a}}\otimes F_{\vec{y}\vec{b}}|\psi}
    $$
    coincide for all strategy producing $p(a,b|x,y)$,
    \item There is a unique finite dimensional real state on $\cA_{\R,\text{POVM}}^{\cI_A,\cO_A}\otimes_{\min}\cA_{\R,\text{POVM}}^{\cI_B,\cO_B}$ that agrees with $p(a,b|x,y)$.
    \end{enumerate}
    Here, $\cA_{\R,\text{POVM}}^{\cI_A,\cO_A}$ is the universal real C* algebra generated by positive contractions $\{e_{xa}:{x\in\cI_A,a\in\cO_A}\}$, subject to the relations $\sum_ae_{xa}=1,\forall x\in\cI_A$, and similarly $\cA_{\R,\text{POVM}}^{\cI_B,\cO_B}$ is generated by $\{f_{yb}:{y\in\cI_B,b\in\cO_B}\}$. A real state $f$ agrees with $p(a,b|x,y)$ whenever $f(e_{xa}\otimes f_{yb})=p(a,b|x,y)$ holds for all $a,b,x,y$.
\end{lemma}

\begin{proof}
    (1)$\Rightarrow$(2): For any finite dimensional real state $f$ that agrees with $p$, its real GNS construction \cite[Theorem 3.3.4]{li2003real} gives a representation on a finite dimensional real Hilbert space, whose matrix representation gives raise to a real strategy which is moment-real. By Proposition \ref{prop:hom}, those $f$ then agrees with all the words of generators, so $f$ is determined on the whole real C* algebra from its real linearity.

    (2)$\Rightarrow$(1): Suppose $S^{(0)},S^{(1)}$ differ in their real parts of moments, define real states $f_0,f_1$ by setting $f_0(e_{\vec{x}\vec{a}}\otimes f_{\vec{y}\vec{b}})=\re\braket{\psi^{(0)}|E^{(0)}_{\vec{x}\vec{a}}\otimes F^{(0)}_{\vec{y}\vec{b}}|\psi^{(0)}}$, $f_1(e_{\vec{x}\vec{a}}\otimes f_{\vec{y}\vec{b}})=\re\braket{\psi^{(1)}|E^{(1)}_{\vec{x}\vec{a}}\otimes F^{(1)}_{\vec{y}\vec{b}}|\psi^{(1)}}$, and extending them by real linearity. Then $f_0,f_1$ are valid real states on $\cA_{\R,\text{POVM}}^{\cI_A,\cO_A}\otimes_{\min}\cA_{\R,\text{POVM}}^{\cI_B,\cO_B}$ but $f_0\neq f_1$.
\end{proof}

We are ready to present the main result of this section.

\begin{theorem}\label{t:reC*}
~
    \begin{enumerate}
        \item If a support-preserving $\tilde S$ is complex self-tested by a correlation $p(a,b|x,y)$, then there is a unique finite-dimensional real state on $\cA_{\R,\text{POVM}}^{\cI_A,\cO_A}\otimes_{\min}\cA_{\R,\text{POVM}}^{\cI_B,\cO_B}$ that agrees with $p(a,b|x,y)$.
        \item Suppose the correlation $p(a,b|x,y)$ is extreme in $C_q$. If there is a unique finite-dimensional real state on $\cA_{\R,\text{POVM}}^{\cI_A,\cO_A}\otimes_{\min}\cA_{\R,\text{POVM}}^{\cI_B,\cO_B}$ that agrees with $p(a,b|x,y)$, then there is a canonical $\tilde S$ such that $\tilde S$ is complex self-tested by correlation $p(a,b|x,y)$.
    \end{enumerate}
    \label{thm:c-realstate}
\end{theorem}

\begin{proof}
    Combine Propositions \ref{prop:hom}, \ref{prop:c-realstate2}, and Lemma \ref{lem:c*=moments}.
\end{proof}

\begin{remark}
For comparison with Theorem \ref{t:reC*} above, consider the results of \cite{paddock2023operatoralgebraic}, with notations slightly modified in accordance with ours (\eg, "centrally supported" in \cite{paddock2023operatoralgebraic} means "support-preserving" here). 
Namely, \cite[Proposition 4.10 and Theorem 4.12]{paddock2023operatoralgebraic} establish the following.
    \begin{enumerate}
        \item If a support-preserving $\tilde S$ is self-tested by a correlation $p(a,b|x,y)$, then there is a unique finite-dimensional state on $\cA_{\text{POVM}}^{\cI_A,\cO_A}\otimes_{\min}\cA_{\text{POVM}}^{\cI_B,\cO_B}$ that agrees with $p(a,b|x,y)$.
        \item Given an correlation $p(a,b|x,y)$ that is extreme in $C_q$. If there is a unique finite-dimensional state on $\cA_{\text{POVM}}^{\cI_A,\cO_A}\otimes_{\min}\cA_{\text{POVM}}^{\cI_B,\cO_B}$ that agrees with $p(a,b|x,y)$, then there is a canonical $\tilde S$ such that $\tilde S$ is self-tested by correlation $p(a,b|x,y)$.
    \end{enumerate}
A reader can now readily compare the concepts of "self-test in terms of complex C* algebra" and "complex self-test in terms of real C* algebra". Note that $\cA_{\text{POVM}}^{\cI,\cO}=\C\otimes_\R \cA_{\R,\text{POVM}}^{\cI,\cO}$; alternatively, $\cA_{\R,\text{POVM}}^{\cI_A,\cO_A}$ is the norm-closed real $*$-subalgebra of $\cA_{\text{POVM}}^{\cI_A,\cO_A}$ generated by the canonical generators of $\cA_{\text{POVM}}^{\cI_A,\cO_A}$. Finally, we remark that Point 2 of Theorem \ref{thm:c-realstate} relies on the assumption that the considered correlation is extreme, similarly to its counterpart in \cite[Theorem 4.12]{paddock2023operatoralgebraic}. This reflects a limitation of the existing proof techniques, and whether this assumption can be relaxed remains an interesting open question.
\end{remark}
\section{Realness of quantum strategies}\label{sec:realstrat}

Section \ref{sec:c-operator-algebraic} indicates that the real parts of higher moments are essential in complex self-testing, and leads our attention to quantum strategies with real moments. An obvious candidate of that is the family of strategies with a real matrix representation. Then the natural question to ask is, are there any other strategies with real higher moments? If the answer is affirmative then it would be a more appropriate definition of `non-complex' quantum strategies in the context of self-testing. 
Here we solve this problem by fully identifying the family of strategies with real higher moments, which we will call `self-conjugate' strategies. 

To investigate whether a quantum strategy has real moments or admits a real matrix representation, one needs to consider the real algebra generated by measurements in the strategy. 
Given $X_1,\dots,X_m\in\mtxc{d}$ let $\Alg_{\R}(X_j\colon j)$ and $\Alg_{\C}(X_j\colon j)$ denote the real unital $*$-algebra and the complex unital $*$-algebra, respectively, generated by $X_1,\dots,X_m$ in $\mtxc{d}$ endowed with the conjugate transpose. The collection $X_1,\dots,X_m$ is \emph{irreducible} if $\Alg_{\C}(X_j\colon j)=\mtxc{d}$. 
In this case, $\Alg_\R(X_j\colon j)$ is isomorphic to $\mtxr{d}$, $\mtxc{d}$ or $\mtxh{d/2}$ as a consequence of Frobenius' theorem \cite{goodearl,brevsar2014introduction}.

Let us also record basic properties of the standard matrix representation of quaternions. Throughout the rest paper denote
\begin{equation}\label{e:J}
J=\begin{pmatrix}0&-1\\1&0\end{pmatrix}
\in\mtxr{2}.
\end{equation}
Let $\HH=\R+\R i+\R j+\R k$ denote the quaternion algebra (see e.g. \cite[Section 1.1]{brevsar2014introduction}), and consider $n\times n$ quaternion matrices $\mtxh{n}=\HH\otimes_\R\mtxr{n}$ as a real $*$-algebra, whose involution is the tensor product of the transpose in $\mtxr{n}$ and the canonical (symplectic) involution in $\HH$.
There is standard $*$-embedding
$$\Phi:\HH\to\mtxc{2},\qquad
\alpha_0+\alpha_1i+\alpha_2j+\alpha_3k\mapsto
\begin{pmatrix}
\alpha_0+\alpha_1i & \alpha_2+\alpha_3i \\
-\alpha_2+\alpha_3i & \alpha_0-\alpha_1i
\end{pmatrix}.
$$
Note that $\ran\Phi$ is, as a real algebra, generated by $i$-multiples of the Pauli matrices $i\sigma_X,i\sigma_Y,i\sigma_Z$.
A direct calculation shows that
\begin{equation}\label{e:quat_conj}
\overline{\Phi(z)}=
J\Phi(z)J^*
\end{equation}
for all $z\in\HH$. Furthermore, $\Phi$ extends to the $*$-embedding of real algebras
$$\Phi_n=\Phi\otimes_\R\Id_{\mtxr{n}}:
\mtxh{n}=\HH\otimes_\R\mtxr{n} \hookrightarrow\mtxc{2}\otimes_\R\mtxr{n}
=\mtxc{2n}.$$
Then $\ran\Phi_n$ generates $\mtxc{2n}$ as a complex algebra, and $\tr X\in\R$ for every $X\in\ran\Phi_n$. By \eqref{e:quat_conj}, we have $\overline{X}=(J\otimes \Id_n)X(J\otimes \Id_n)^*$ for all $X\in\ran\Phi_n$. 
This is a distinguishing feature of quaternionic matrices (as opposed to real and complex matrices): namely, in their irreducible representation on a Hilbert space, entry-wise complex conjugation coincides with conjugation by a unitary.

\begin{proposition}\label{p:irr_real}
For an irreducible collection $X_1,\dots,X_m\in\mtxc{d}$, consider the following statements:
\begin{enumerate}
\item there is $U\in\U_d(\C)$ such that $UX_jU^*\in\mtxr{d}$ for $j=1,\dots,m$;
\item there is $U\in\U_d(\C)$ such that $UX_jU^*=\overline{X_j}$ for $j=1,\dots,m$;
\item $\tr X\in \R$ for every product $X$ of $X_1,\dots,X_m$;
\item $\Alg_{\R}(X_j\colon j)\cap\C I=\R I$;
\item $\Alg_{\R}(X_j\colon j)\neq \mtxc{d}$.
\end{enumerate}
Then (1)$\Rightarrow$(2)$\Leftrightarrow$(3)$\Leftrightarrow$(4)$\Leftrightarrow$(5). 
\\
If $d$ is odd, or $d=2$ and $X_j$ are hermitian, or $d\in\{4,6\}$ and $m\le3$ and $X_j$ are projections, then (1)$\Leftrightarrow$(2)$\Leftrightarrow$(3)$\Leftrightarrow$(4)$\Leftrightarrow$(5).
\end{proposition}

\begin{proof}
The implications (2)$\Rightarrow$(3)$\Rightarrow$(4)$\Rightarrow$(5) are straightforward. Also, $UXU^*\in\mtxr{d}$ for $U\in\U_d(\C)$ implies $\overline{X}=(\overline{U}^*U)X(\overline{U}^*U)^*$ and $\overline{U}^*U\in\U_d(\C)$, so (1)$\Rightarrow$(2) holds.

Now assume (5) holds. Since $\Alg_\R(X_j\colon j)$ is closed under the conjugate transpose, it is a semisimple real algebra. Furthermore, it is a simple real algebra since $X_j$ are irreducible. Therefore $\Alg_\R(X_j\colon j)$ is isomorphic to one of the $\mtxr{n},\mtxc{n},\mtxh{n}$ for some $n\in\N$ by \cite[Corollary 2.69]{brevsar2014introduction}. Moreover, since the involution on $\Alg_\R(X_j\colon j)$ is a restriction of the conjugate transpose and is therefore positive, the $*$-algebra $\Alg_\R(X_j\colon j)$ is $*$-isomorphic to one of the real $*$-algebras $\mtxr{n},\mtxc{n},\mtxh{n}$ with their canonical standard involutions for some $n\in\N$ by \cite[Theorem 1.2]{NCreal1976}. Note that there is a canonical surjective $*$-homomorphism of complex $*$-algebras
\begin{equation}\label{e:surj}
\C\otimes_\R \Alg_\R(X_j\colon j)\to \Alg_\C(X_j\colon j)=\mtxc{d},
\end{equation}
and
$$
\C\otimes_\R\mtxr{n} \cong \mtxc{n},\quad
\C\otimes_\R\mtxc{n} \cong \mtxc{n}\times\mtxc{n},\quad
\C\otimes_\R\mtxh{n} \cong \mtxc{2n}.
$$
Suppose $\Alg_\R(X_j\colon j)\cong\mtxc{n}$. Surjectivity of the homomorphism \eqref{e:surj} implies $n=d$, and therefore $\Alg_\R(X_j\colon j)=\mtxc{d}$, which contradicts (5). Therefore $\Alg_\R(X_j\colon j)$ is $*$-isomorphic to either $\mtxr{n}$ or $\mtxh{n}$. Since $\mtxc{n}$ and $\mtxc{2n}$ are simple algebras, the surjective homomorphism \eqref{e:surj} is also injective, and therefore either $\Alg_\R(X_j\colon j)\cong \mtxr{d}$, or $d$ is even and $\Alg_\R(X_j\colon j)\cong \mtxh{d/2}$. 

In the first case, there is a $*$-isomorphism of real algebras $\Psi:\Alg_\R(X_j\colon j)\to \mtxr{d}$. Then $\Id_\C\otimes_\R\Psi:\mtxc{d}\to\mtxc{d}$ is an automorphism. By the Skolem-Noether theorem \cite[Theorem 1.30]{brevsar2014introduction} there exists $V\in\operatorname{GL}_d(\C)$ such that
$VX_jV^{-1} = \Psi(X_j)$ for $j=1,\dots,m$. Since $\Psi$ is a $*$-homomorphism,
$$V^{-*}X_jV^*=\left(VX_j^*V^*\right)^*
=\Psi(X_j^*)^*=\Psi(X_j)=VX_jV^{-1}$$
and therefore $X_jV^*V=V^*VX_j$ for $j=1,\dots,m$. Since $X_j$ are irreducible, it follows that $V^*V=\alpha I$ for some nonzero scalar $\alpha$. Clearly $\alpha>0$. Then $U=\frac{1}{\sqrt{\alpha}}V\in \U_d(\C)$ satisfies (1).

Now consider the second case. Then $d$ is even and there is a $*$-isomorphism of real algebras $\Psi:\Alg_\R(X_j\colon j)\to \mtxh{d/2}$. Then $\Id_\C\otimes_\R\Psi:\mtxc{d}\to\mtxc{d}$ is again an automorphism. By the Skolem-Noether theorem \cite[Theorem 1.30]{brevsar2014introduction} there exists $V\in\operatorname{GL}_d(\C)$ such that
$VX_jV^{-1} = \Phi_n(\Psi(X_j))$ for $j=1,\dots,m$, where $\Phi_n:\mtxh{d/2}\hookrightarrow\mtxc{d}$ is the $*$-embedding from Section \ref{sec:quat_prelim}.
Since $\Phi_n\circ\Psi$ is a $*$-homomorphism,
$$V^{-*}X_jV^*=\left(VX_j^*V^*\right)^*
=\Big((\Phi_n\circ\Psi)(X_j^*)\Big)^*=(\Phi_n\circ\Psi)(X_j)=VX_jV^{-1}$$
and therefore $X_jV^*V=V^*VX_j$ for $j=1,\dots,m$. Since $X_j$ are irreducible, it follows that $V^*V=\alpha I$ for some nonzero scalar $\alpha$. Clearly $\alpha>0$. Then $W=\frac{1}{\sqrt{\alpha}}V\in \U_d(\C)$ satisfies $WX_jW^* = \Phi_n(\Psi(X_j))$ for $j=1,\dots,m$. By  \eqref{e:quat_conj},
$$\overline{WX_jW^*}=
\overline{\Phi_n(\Psi(X_j))}
=(J\otimes I_{d/2})\Phi_n(\Psi(X_j))(J\otimes I_{d/2})^*
=(J\otimes I_{d/2})WX_jW^*(J\otimes I_{d/2})^*
$$
and therefore $\overline{X_j}=UX_jU^*$
for $U=W^\intercal(J\otimes I_{d/2})W\in\U_d(\C)$, so (2) holds.

Finally, notice that the first case $\Alg_\R(X_j\colon j)\cong \mtxr{d}$ is the only possibility whenever $d$ is odd, or if $d=2$ and $X_j$ are hermitian matrices (since $\HH$ is not generated by hermitian elements), or if $d\in\{4,6\}$, $m\le 3$ and $X_j$ are projections by Proposition \ref{p:quat} (which we will formally introduce later).
\end{proof}

Recall that a finite-dimensional strategy $S$ is irreducible if the $\{E_{xa}\}_{x,a}$ generate $L(\cH_A)$ and the $\{F_{yb}\}_{y,b}$ generate $L(\cH_B)$ as complex algebras. 

\begin{definition}
    The strategy $S$ is:
\begin{enumerate}
\item \emph{(Schmidt) real} if some (Schmidt) matrix representation of $S$ is real;
\item \emph{(Schmidt) self-conjugate} if for some/all (Schmidt) basis there exist local unitaries $U_A,U_B$ such that 
    \begin{align*}
        U_AE_{xa}U_A^*=\overline{E_{xa}},~U_BF_{yb}U_B^*=\overline{F_{yb}},~U_A\otimes U_B\ket{\psi}={\ket{\overline\psi}}
    \end{align*}
    holds for all $x,y,a,b$.
\item \emph{moment-real} if $\bra{\psi}E\otimes F\ket{\psi}\in\R$ for all words $E_{\vec{x}\vec{a}}$ of POVM operators $\{E_{xa}\}$ and words $F_{\vec{y}\vec{b}}$ of POVM operators $\{F_{yb}\}$.
\end{enumerate}
    \label{def:c-self-conj}
\end{definition}

\begin{theorem}\label{thm:self-conj}
    For an irreducible strategy $S$ with local Hilbert spaces $\cH_A$ and $\cH_B$, consider the following statements:
\begin{enumerate}
\item $S$ is Schmidt real;
\item $S$ is real;
\item $S$ is Schmidt self-conjugate;
\item $S$ is self-conjugate;
\item $S$ is moment-real.
\end{enumerate}
Then (1)$\Leftrightarrow$(2)$\Rightarrow$(3)$\Leftrightarrow$(4)$\Leftrightarrow$(5). 
\\
If on Alice's side and Bob's side, at least one of the conditions
\begin{itemize}
\item local dimension is $2$ or odd,
\item local dimension is $4$ or $6$, there are at most three inputs, and measurements are binary and projective,
\end{itemize}
is fulfilled, then (1)$\Leftrightarrow$(2)$\Leftrightarrow$(3)$\Leftrightarrow$(4)$\Leftrightarrow$(5).
\end{theorem}

\begin{proof}
    The equivalences (1)$\Leftrightarrow$(2) and (3)$\Leftrightarrow$(4) follow from the existence of the singular value decomposition for real matrices, and the implications (2)$\Rightarrow$(4)$\Rightarrow$(5) are straightforward. 
Let $d_A=\dim\cH_A$ and $d_B=\dim\cH_B$.

For the sake of contradiction, suppose that $S$ is moment real but not self-conjugate. Assume that the conditions in Definition \ref{def:c-self-conj} fail for $\{E_{xa}\}$. Then by Proposition \ref{p:irr_real} $\{E_{xa}\}$ generate $L(\cH_A)$ as a real algebra. Let $\ket{\psi}=\sum_{i=1}^r\ket{u_i}\ket{v_i}$ for linearly independent $\ket{u_i}\in\cH_A$ and linearly independent $\ket{v_i}\in\cH_B$. Note that
\begin{equation}\label{e:higher_cx}
\bra{\psi}E\otimes F\ket{\psi}=\sum_{i,j=1}^r
\bra{u_i}E\ket{u_j}\cdot
\bra{v_i}F\ket{v_j}
\end{equation}
for $E\in L(\cH_A)$ and $F\in L(\cH_B)$. Since $\{F_{yb}\}$ are irreducible, there exists a word $F$ of 
$\{F_{yb}:y,b\}$ such that not all $\bra{v_i}F\ket{v_j}$ are 0. In particular, $\bra{v_{i_0}}F\ket{v_{j_0}}\neq0$ for some $i_0,j_0$. Since $\{E_{xa}\}$ generate $L(\cH_A)$ as a real algebra, there is a real combination $E=\sum_k\alpha_kE_k$ of words $E_k$ of $E_{xa}$ such that
$$
\bra{u_i}E\ket{u_j}=\left\{
\begin{array}{cc}
i\overline{\bra{v_{i_0}}F\ket{v_{j_0}}} & \text{if }i=i_0,j=j_0 \\
0 & \text{otherwise}.
\end{array}\right.
$$
Therefore
$$\sum_k\alpha_k\bra{\psi}E_k\otimes F\ket{\psi}
=\bra{\psi}E\otimes F\ket{\psi}\notin\R$$
by Eq.~\eqref{e:higher_cx}, and so $\bra{\psi}E_k\otimes F\ket{\psi}\notin\R$ for some $k$, which contradicts $S$ being moment real. 

Therefore there are $U_A\in\U_{d_A}(\C)$ and $U_B\in\U_{d_B}(\C)$ such that
\begin{align*}
    U_AE_{xa}U_A^*=\overline{E_{xa}},~U_BF_{yb}U_B^*=\overline{F_{yb}}
\end{align*}
for all $x,y,a,b$. Denote $\ket{\psi'}=U_A^*\otimes U_B^*\ket{\overline\psi}$.
Clearly,
\begin{equation}\label{e:psipsi}
\braket{\psi'|E\otimes F|\psi'}=\braket{\overline{\psi}|U(E\otimes F)U^*|\overline{\psi})}=\overline{\braket{\psi|E\otimes F|\psi}}=\braket{\psi|E\otimes F|\psi}
\end{equation}
for all words $E$ of $E_{xa}$ and words $F$ of $F_{yb}$. Since both sides of Eq.~\eqref{e:psipsi} are complex linear in $E\otimes F$, and $\Alg_\C(E_{xa}\colon x,a)\otimes\Alg_\C(F_{yb}\colon y,b)=L(\cH_A\otimes\cH_B)$, it furthermore follows that $\tr(\proj{\psi'}T)=\braket{\psi'|T|\psi'}=\braket{\psi|T|\psi}=\tr(\proj{\psi}T)$ for all $T\in L(\cH_A\otimes\cH_B)$.
Then
$\proj{\psi}=\proj{\psi'}$, and so $\ket{\psi'}= \alpha \ket{\psi}$ for some phase $\alpha\in\C$ of modulus 1. Therefore, we have
$$\overline{E_{xa}}=(\alpha U_A)E_{xa}(\alpha U_A)^*,\quad
\overline{F_{yb}}=U_BF_{yb}U_B^*,\quad
\ket{\overline{{\psi}}}=U_A\otimes U_B\ket{\psi'}=(\alpha U_A)\otimes U_B\ket{\psi}$$
for unitaries $\alpha U_A$ and $U_B$. Thus, (5)$\Leftrightarrow$(4).

Finally, assume that (4) holds, and that at least one of the exceptional conditions is fulfilled on Alice's and on Bob's side. By Proposition \ref{p:irr_real} there in particular exist orthonormal bases $\cB_A$ and $\cB_B$ relative to which the measurements in $S$ are given by real matrices $E_{xa}\in\mtxr{d_A}$ and $F_{yb}\in\mtxr{d_B}$. Since $S$ is irreducible, $E_{xa}$ and $F_{yb}$ generate $\mtxr{d_A}$ and $\mtxr{d_B}$ as real algebras. Therefore by (5),
$$\braket{\psi|A\otimes B|\psi}\in\R$$
for all $A\in \mtxr{d_A}$ and $B\in\mtxr{d_B}$. Write $\ket{\psi}=\sum_{i,j}\alpha_{ij}\ket
{ij}$ relative to bases $\cB_A$ and $\cB_B$; then
$$\alpha_{ij}\overline{\alpha_{k\ell}} = \braket{\psi|(\ket{k}\!\bra{i}\otimes\ket{j}\!\bra{\ell})|\psi}\in\R$$
for all $i,k=1,\dots,d_A$ and $j,\ell=1,\dots,d_B$. Therefore arguments of $\{\alpha_{ij}\}$ coincide, so there is $\zeta\in\C$ of modulus 1 such that $\zeta \{\alpha_{ij}\}\in\R^{d_A\times d_B}$. Therefore $S$ is a real strategy, with corresponding orthonormal bases $\zeta\cB_A$ and $\cB_B$, and hence (2) holds.
\end{proof}

Let us point out that there exist strategies that are moment-real but not real.

\begin{example}\label{ex:quat_strategy}
Let $d\ge4$ be even. By Proposition \ref{p:quat}, there exists an irreducible collection of projections $P_1,\dots,P_4\in\mtxc{d}$ such that $\tr(P)\in\R$ for every product $P$ of $P_1,\dots,P_4$, and there is no $U\in\U_d(\C)$ such that $UP_jU^*\in\mtxr{d}$ for all $j=1,\dots,4$. 
Namely, for $P_j$ one can take any projective generators of $\mtxh{d/2}$ within $\mtxc{d}$ (if $d\ge8$, three projections suffice).
Let
$$S=\left(
\ket{\phi_d},\{P_j,\Id-P_j\}_{j=1}^4,\{P_j^\intercal,\Id-P_j^\intercal\}_{j=1}^4
\right)$$
where $\ket{\phi_d}=\frac{1}{\sqrt{d}}\sum_{i=1}^d\ket{ii}$ is the canonical maximally entangled state. Then the strategy $S$ is self-conjugate but not real.
\end{example}

As seen in the above arguments, the distinction between real and self-conjugate strategies essentially boils down to the fact that realness of trace cannot distinguish between real matrices and quaternion matrices. Nevertheless, tracial identities distinguish between $\mtxr{m}$ and $\mtxh{n}$ for all $m,n$ by \cite[Corollary 2.5.12 and Remark 2.5.1]{row80} and \cite[Proposition 2.3]{ksv18}.

Given the result of Theorem \ref{thm:self-conj}, we shall call a strategy $S$ \emph{non-real} if $S$ is not real (or not Schmidt real). We shall call $S$ \emph{complex} if $S$ is not self-conjugate (or not moment-real). We will show that it is exactly the class of complex strategies that cannot be self-tested (in the standard sense). 

\begin{theorem}\label{p:notss}
    Let $\tilde S=(\ket{\tilde\psi},\{\tilde{E}_{xa}\},\{\tilde{F}_{yb}\})$ be a full-rank complex PVM strategy. Then $\tilde S$ is not self-tested.
\end{theorem}

\begin{proof}
Not all moments of $\tilde S$ are real, the strategy $\tilde S$ and its complex conjugate give rise to distinct states on $\cA_{\C,\text{POVM}}^{\cI_A,\cO_A}\otimes_{\min}\cA_{\C,\text{POVM}}^{\cI_B,\cO_B}$. Thus, the correlation of $\tilde S$ is not an abstract self-test, and thus not a self-test by \cite[Proposition 4.10]{paddock2023operatoralgebraic}.
\end{proof}

\begin{lemma}
    Let $S$ be a complex strategy and $S_R$ be its real simulation. Then $S\xhookrightarrow{}S_R$ does not hold.
\end{lemma}

\begin{proof}
    If $S\xhookrightarrow{}S_R$ then $S$ and $S_R$ have the same moments. But $S_R$ is moment-real, while $S$ is not.
\end{proof}

In the end of this section, we point out that for a full-rank, PVM strategy to be complex self-tested, it cannot be merely `one-sided real'.

\begin{theorem}\label{t:notss}
    Let $\tilde S=(\ket{\tilde\psi},\{\tilde{E}_{xa}\},\{\tilde{F}_{yb}\})$ be a full-rank PVM strategy. If there exist a basis where
    \begin{itemize}
        \item $\ket{\tilde\psi}$ has real matrix representation,
        \item $\{\tilde E_{xa}\}$ has real matrix representation,
        \item at least one of $\{\tilde F_{yb}\}$ has no real matrix representation,
    \end{itemize}
    then $\tilde S$ is not complex self-tested.
\end{theorem}

\begin{proof}
    We prove this by showing $\re\tilde S:=(\ket{\tilde{\psi}},\{\tilde{E}_{xa}\},\{\re\tilde{F}_{yb}\})$ produces the same correlation as $\tilde S$, but cannot be complex local dilated to $\tilde S$. First, we note that $\re\tilde{F}_{yb}=1/2(\tilde{F}_{yb}+\overline{\tilde{F}_{yb}})$, a convex combination of POVMs. So $\{\re\tilde{F}_{yb}\}_b$ is a valid POVM.
    
    To show that $\re\tilde S$ produces the same correlation as $\tilde S$, notice that $\operatorname{im}\tilde{F}_{yb}$ is anti-symmetric. Therefore $\braket{\tilde\psi|\tilde{E}_{xa}\otimes\operatorname{im}\tilde{F}_{yb}|\tilde\psi}=0$ for all $a,b,x,y$. Then $\braket{\tilde\psi|\tilde{E}_{xa}\otimes\tilde{F}_{yb}|\tilde\psi}=\braket{\tilde\psi|\tilde{E}_{xa}\otimes\re\tilde{F}_{yb}|\tilde\psi}$.

    To show that $\re\tilde S$ cannot be complex local dilated to $\tilde S$, we first prove that for any measurement $F_{yb}$, $\{\re{F}_{yb}\}$ is a PVM if and only if ${F}_{yb}=\overline{ F_{yb}}$ . We have that
    \begin{align*}
        \re{F}_{yb}^2=\frac{1}{4}({F}_{yb}+\overline{\tilde{F}_{yb}}+\overline{\tilde{F}_{yb}}{F}_{yb}+{F}_{yb}\overline{{F}_{yb}}).
    \end{align*}
    So $\re{F}_{yb}$ being projection ($(\re{F}_{yb})^2=\re{F}_{yb}$) is equivalent to 
    \begin{align*}
        &\overline{{F}_{yb}}{F}_{yb}+{F}_{yb}\overline{{F}_{yb}}=\overline{{F}_{yb}}+{F}_{yb}\\
        \Leftrightarrow &{F}_{yb}\overline{{F}_{yb}}{F}_{yb}={F}_{yb}.
    \end{align*}
    Also note that both ${F}_{yb},\overline{{F}_{yb}}$ are projections of the same rank. So this implies ${F}_{yb}=\overline{{F}_{yb}}$.

    Hence given the assumption, $\re\tilde S$ is not a PVM. Therefore $\re\tilde S\hookrightarrow_{\C}\tilde S$ does not hold, because complex local dilation preserves projectivity (Proposition \ref{prop:cinvariantProj}).
\end{proof}
\section{A quaternion middle ground possibility}

Proposition \ref{p:notss} shows that complex (not self-conjugate) strategies cannot be self-tested. On the other hand, it is known that every real projective set of measurements can be embedded into a real strategy that is self-tested \cite{Chen2024}. Given the existence of self-conjugate but not real strategies, it is natural to ask whether there exist self-tests within this middle ground? In this section we give an affirmative answer to this question, and the construction arises from quaternions along with an extension of CHSH inequality. 

\subsection{A self-test involving quaternions}\label{sec:quat_selftest}

The real algebra $\mtxh{2}$ is generated by hermitian unitaries
$$
h_1=\begin{pmatrix}1&0\\0&-1\end{pmatrix},\ 
h_2=\begin{pmatrix}0&1\\1&0\end{pmatrix},\ 
h_3=\begin{pmatrix}0&-i\\i&0\end{pmatrix},\ 
h_4=\begin{pmatrix}0&-j\\j&0\end{pmatrix}.
$$
The elements $h_1,\dots,h_4$ pairwise anticommute. Under the standard $*$-embedding $\Phi_2:\mtxh{2}\hookrightarrow\mtxc{4}$ from Section \ref{sec:pre}, they are represented by $X_1,\dots,X_4$ given as
\begin{align*}
X_1=\begin{pmatrix}1&0&0&0\\0&1&0&0\\0&0&-1&0\\0&0&0&-1\end{pmatrix}=\sigma_Z\otimes \Id,
\ 
X_2=\begin{pmatrix}0&0&1&0\\0&0&0&1\\1&0&0&0\\0&1&0&0\end{pmatrix}=\sigma_X\otimes\Id,\\
\ 
X_3=\begin{pmatrix}0&0&-i&0\\0&0&0&i\\i&0&0&0\\0&-i&0&0\end{pmatrix}=\sigma_Y\otimes\sigma_Z,
\ 
X_4=\begin{pmatrix}0&0&0&-1\\0&0&1&0\\0&1&0&0\\-1&0&0&0\end{pmatrix}=\sigma_Y\otimes\sigma_Y.
\end{align*}

The following statement is a special case of a well-known Clifford algebra formalism \cite{porteous}.

\begin{proposition}\label{p:4anticomm}
If $a_1,\dots,a_4$ are four pairwise anticommuting hermitian unitaries on a complex Hilbert space $\cH$, then the unital real subalgebra generated by them is isomorphic to $\mtxh{2}$, and there exists a unitary $U:\cH\to \C^4\otimes\cK$ for some Hilbert space $\cK$, such that
$$Ua_\ell U^*=X_\ell\otimes \Id\qquad \text{for }\ell=1,\dots,4.$$
\end{proposition}

\begin{proof}
For $\xi=(\xi_1,\xi_2,\xi_3,\xi_4)\in\R^4$,
$$(\xi_1a_1+\xi_2a_2+\xi_3a_3+\xi_4a_4)^2=\|\xi\|^2 1.$$
In particular, $a_1,\dots,a_4$ are linearly independent. Let $V$ be the real subspace of $\cA$ spanned by $a_1,\dots,a_4$. Then $v^2=\|v\|^2 1$ for all $v\in V$. By the universal property of Clifford algebras \cite[Theorem 15.13]{porteous}, the unital real algebra $\Alg_\R(a_1,\dots,a_4)$ is a homomorphic image of the real Clifford algebra ${\rm Cl}_{4,0}(\R)$, which is isomorphic to $\mtxh{2}$ \cite[Table 15.27]{porteous}. The latter algebra is simple (\ie, has no nonzero proper ideals), so $\Alg_\R(a_1,\dots,a_4)\cong \mtxh{2}$ via the map $a_\ell\mapsto h_\ell$. Finally, the complexification of $\mtxh{2}$ is the matrix algebra $\C\otimes_{\R}\mtxh{2}\cong \mtxc{4}$, whose $*$-embeddings into $\cB(\cH)$ are all unitarily equivalent. Thus, $\cH$ factors as $\C^4\otimes \cK$ for some subspace $\cK\subset\cH$, and there is a unitary $U:\cH\to \C^4\otimes\cK$ such that $Ua_\ell U^*=X_\ell\otimes \Id$ for $\ell=1,\dots,4$.
\end{proof}

Let $\ket{\phi_4}=\frac12(\ket{00}+\ket{11}+\ket{22}+\ket{33})\in\C^4\otimes \C^4$ be the canonical maximally entangled state of local dimension 4. The following statement can be verified by a direct calculation.

\begin{lemma}\label{l:eigvec}
The largest eigenvalue of $X_1\otimes X_1+X_2\otimes X_2-X_3\otimes X_3+X_4\otimes X_4$ is 4, and the corresponding eigenspace is spanned by $\ket{\phi_4}$. 
\end{lemma}

We also record a well-known fact about anticommuting unitaries.

\begin{lemma}\label{l:basicanticomm}
Let $\theta\in\R$. If $X$ and $Y$ are anticommuting hermitian unitaries, then so are
$\cos\theta\,X+\sin\theta\,Y$ and $\sin\theta\,X-\cos\theta\,Y$.
\end{lemma}

\def\chsh{{\rm CHSH}}
For a bipartite state $\ket{\psi}$ and observables $A_i,B_j$ denote
$$\chsh\big(\ket{\psi};A_1,A_2;B_1,B_2\big)=\bra{\psi}
(A_1+A_2)\otimes B_1+(A_1-A_2)\otimes B_2
\ket{\psi},$$
the CHSH inequality expression.
For $\ell\neq m \in\{1,\dots,4\}$ denote
$$Y_{\ell m,\pm}=\frac{1}{\sqrt{2}}(X_\ell\pm X_m)^\intercal.$$
Then for every pair $\ell\neq m$, 
$$\chsh\big(\ket{\phi_4};X_\ell,X_m;Y_{\ell m,+},Y_{\ell m,-}\big)=2\sqrt{2},$$
\ie, $(\ket{\phi_4},\{X_\ell,X_m\},\{Y_{\ell m,+},Y_{\ell m,-}\})$ is an optimal strategy for the CHSH inequality. Now consider the bipartite strategy
\begin{equation}\label{e:quaternion_strat}
S=\left(
\ket{\phi_4},\{X_1,X_2,X_3,X_4\},\{Y_{\ell m,\pm}\colon 1\le \ell<m\le 4\}
\right).
\end{equation}
Let $J$ be as in Eq.~\eqref{e:J}.
Since $(\Id_2\otimes J\otimes \Id_2\otimes J)\ket{\phi_4}=\ket{\phi_4}$ and the measurement algebras of $S$ are isomorphic to $\mtxh{2}$, 
the strategy $S$ is unitarily equivalent to its complex conjugate via local unitaries $\Id_2\otimes J$ (on both sides) by Eq.~\eqref{e:quat_conj}.
That is, $S$ is Schmidt self-conjugate, and non-real by Example \ref{ex:quat_strategy}.

Moreover, below we show that $S$ is self-tested (in the original, real sense) by a Bell inequality (without additional assumptions on the comparing strategies; see \cite[Theorem B.1]{baptista2023mathematicalfoundationselftestinglifting}). To the best of our knowledge, this is the first known real self-test of a non-real strategy. 

Like many other self-testing results in the literature, the following theorem exploits the well-known properties of the CHSH inequality. In particular, the Bell inequality self-testing $S$ is a sum of six CHSH inequalities, inspired by \cite{bowles_self-testing_2018}, where a sum of three CHSH inequalities is used to certify the Pauli observables. More generally, the role of Clifford algebras in identifying quantum measurements has been long recognized \cite{Tsirelson1987QuantumAO,SlofstraLower}.

\begin{theorem}\label{t:quaternion}
The strategy $S$ \eqref{e:quaternion_strat} is self-tested by the Bell inequality
\begin{equation}\label{e:4bell}
\sum_{1\le\ell<m\le 4}
\chsh\big(\ket{\phi_4};X_\ell,X_m;Y_{\ell m,+},Y_{\ell m,-}\big)\le 12\sqrt{2}.
\end{equation}
\end{theorem}

\begin{proof}
Let
$$S'=\left(
\ket{\psi},\{X_1',X_2',X_3',X_4'\},\{Y_{\ell m,\pm}'\colon 1\le \ell<m\le 4\}
\right)$$
be another strategy (whose measurements are given by hermitian unitaries) on Hilbert subspaces $\cH_A$ and $\cH_B$ that attains equality in Ineq.~\eqref{e:4bell} (by \cite[Theorem B.1]{baptista2023mathematicalfoundationselftestinglifting}, it suffices to restrict to such strategies in order to establish an assumption-free self-test).
Since $\chsh\le 2\sqrt{2}$ on observables, it follows that
\begin{equation}\label{e:chshstrat}
\chsh\big(\ket{\psi};X_\ell',X_m';Y_{\ell m,+}',Y_{\ell m,-}'\big)=2\sqrt{2}
\end{equation}
for all $\ell<m$.
By the self-testing feature of the CHSH inequality (\eg, \cite[Section 4]{SB}), the strategies $S_{\ell m}=(\ket{\phi_4},\{X_\ell,X_m\},\{Y_{\ell m,+},Y_{\ell m,-}\})$ and $S_{\ell m}'=(\ket{\psi},\{X_\ell',X_m'\},\{Y_{\ell m,+}',Y_{\ell m,-}'\})$ give rise to the same correlation, for all $\ell<m$.

Since $\frac{1}{\sqrt{2}}(X_\ell'\pm X_m')$ and $Y_{\ell m,\pm}'$ are hermitian unitaries (the former one by Lemma \ref{l:basicanticomm}),
\begin{align*}
\left\| \tfrac{1}{\sqrt{2}}(X_\ell'\pm X_m')\otimes \Id\ket{\psi}
-\Id\otimes Y_{\ell m,\pm}'\ket{\psi}\right\|^2 
&=2-2\bra{\psi}\tfrac{1}{\sqrt{2}}(X_\ell'\pm X_m')\otimes Y_{\ell m,\pm}'\ket{\psi}\\
&=2-2\bra{\phi_4}\tfrac{1}{\sqrt{2}}(X_\ell\pm X_m)\otimes Y_{\ell m,\pm}\ket{\phi_4}\\
&=\left\| \tfrac{1}{\sqrt{2}}(X_\ell\pm X_m)\otimes \Id\ket{\phi_4}
-\Id\otimes Y_{\ell m,\pm}\ket{\phi}_4\right\|^2\\
&=\left\| \big(Y_{\ell m,\pm}^\intercal\otimes \Id-\Id\otimes Y_{\ell m,\pm}\big)\ket{\phi_4}\right\|^2\\
&=\tfrac12 \tr(Y_{\ell m,\pm}-Y_{\ell m,\pm})=0
\end{align*}
due to $S_{\ell m}'$ and $S_{\ell m}$ having the same correlations, and the tracial property of $\ket{\phi_4}$.
Thus,
\begin{equation}\label{e:leftright}
\frac{1}{\sqrt{2}}(X_\ell'\pm X_m')\otimes \Id\ket{\psi}
=\Id\otimes Y_{\ell m,\pm}'\ket{\psi}
\qquad\text{for all }\ell<m.
\end{equation}
In particular, the observables in $S$ preserve the support of $\ket{\psi}$. By \cite[Proposition 3.8 and Theorem 3.9]{baptista2023mathematicalfoundationselftestinglifting}, we can replace $S$ with a suitable local dilation, and consequently assume that $\supp_A\ket{\psi}=\cH_A$ and $\supp_B\ket{\psi}=\cH_B$ (\ie, $\ket{\psi}$ is of full Schmidt rank).  

By the self-testing feature of the CHSH inequality \cite[Section 4.2]{SB} for the strategy \eqref{e:chshstrat}, we have
\begin{equation}\label{e:anticomm_from_chsh}
X_\ell'X_m'+X_m'X_\ell'=0,\qquad
Y_{\ell m,+}'Y_{\ell m,-}'+Y_{\ell m,-}'Y_{\ell m,+}'=0
\end{equation}
for all $\ell<m$.
Define $X_\ell'':=\frac{1}{\sqrt{2}}(Y_{\ell 4,\pm}'+Y_{\ell 4,-}')^\intercal$ for $\ell=1,2,3$ and 
$X_4'':=\frac{1}{\sqrt{2}}(Y_{14,+}'-Y_{14,-}')^\intercal$.
Using Eq.~\eqref{e:leftright}, one derives that
\begin{equation}\label{e:lincomb}
Y_{\ell m,\pm}'=\frac{1}{\sqrt{2}}(X_\ell'\pm X_m')^\intercal\quad 
\text{for all }\ell<m.
\end{equation}
Thus, $Y_{\ell m,\pm}'^\intercal$ are linear combinations of $X_1'',\dots,X_4''$ in the precisely the same way as $Y_{\ell m,\pm}^\intercal$ are linear combinations of the $X_\ell$.
Furthermore, $X_1'',\dots,X_4''$ are anticommuting hermitian unitaries by Eqs. \eqref{e:anticomm_from_chsh} and \eqref{e:lincomb}, and Lemma \ref{l:basicanticomm}.
Consequently,
$$\Alg_\R(X_1',\dots,X_4')\cong \mtxh{2}\cong \Alg_\R(X_1'',\dots,X_4'')$$
by Proposition \ref{p:4anticomm}.
Thus, there exist unitaries $U_A:\cH_A\to \C^4\otimes\cH_{A'}$ and $U_B:\cH_B\to \C^4\otimes\cH_{B'}$ such that
$$U_AX_\ell'U_A^*=X_\ell\otimes \Id_{\cH_{A'}},\quad 
U_BX_\ell''U_B^*=X_\ell\otimes \Id_{\cH_{B'}}
\qquad \text{for }1\le\ell\le4.$$
Furthermore,
$$U_BY_{\ell m,\pm}'U_B^*=Y_{\ell m,\pm}\otimes \Id_{\cH_{B'}}\qquad \text{for }1\le\ell<m\le4.$$
Finally, with a slight abuse of tensor ordering in the first line,
\begin{align*}
&\bra{\psi}(U_A\otimes U_B)^*\left(\big(
 X_1\otimes X_1+X_2\otimes X_2-X_3\otimes X_3+X_4\otimes X_4
\big)\otimes \Id_{\cH_{A'}\otimes \cH_{B'}}\right)(U_A\otimes U_B)\ket{\psi}\\
=\,&\bra{\psi}\big(
X_1'\otimes X_1''+X_2'\otimes X_2''-X_3'\otimes X_3''+X_4'\otimes X_4''
\big)\ket{\psi}\\
=\,&\bra{\phi_4} X_1\otimes X_1+X_2\otimes X_2-X_3\otimes X_3+X_4\otimes X_4\ket{\phi_4}
\end{align*}
since the correlations of $S_{\ell m}'$ and $S_{\ell m}$ coincide, and Eq.~\eqref{e:lincomb}.
Lemma \ref{l:eigvec} then implies that
$$(U_A\otimes U_B)\ket{\psi}=\ket{\phi_4}\otimes\ket\aux$$
for some $\ket{\aux}\in\cH_{A'}\otimes \cH_{B'}$.
Thus, $S$ is a local dilation of $S'$.
\end{proof}

\begin{remark}
The self-test in Theorem \ref{t:quaternion} is robust: if a strategy $S'$ attains $12\sqrt{2}-\varepsilon$ (for sufficiently small $\varepsilon>0$) on the left-hand side of the Bell inequality \eqref{e:4bell}, then $S$ is an $\cO(\sqrt{\varepsilon})$-local dilation of $S'$ \cite[Definition 2.5]{baptista2023mathematicalfoundationselftestinglifting}. 
For a wider scope of stability and rigidity results for Clifford algebras, see \cite[Section VI]{SlofstraLower} and \cite[Section 13.4]{RegevVidick}.
Without going into technical details, let us indicate step-by-step how the robustness estimates appear in the proof of Theorem \ref{t:quaternion}:
\begin{enumerate}[label=(\arabic*)]
\item If the left-hand side of Ineq.~\eqref{e:4bell} at $S'$ is $12\sqrt{2}-\varepsilon$, then the correlations of $S_{\ell m}$ and $S_{\ell m}'$ differ up to $\cO(\varepsilon)$, for all $\ell<m$.
\item Then, Eq.~\eqref{e:leftright} holds up to $\cO(\sqrt{\varepsilon})$; hence, after one replaces $S$ with a strategy with a fully supported state, one may still assume that the correlations of $S_{\ell m}$ and $S_{\ell m}'$ differ up to $\cO(\varepsilon)$ by \cite[Lemma 3.3 and Proposition 3.4]{baptista2023mathematicalfoundationselftestinglifting}.
\item The anticommutation relations in Eq.~\eqref{e:anticomm_from_chsh} then hold on $\ket{\psi}$ up to $\cO(\sqrt{\varepsilon})$ by \cite[Section 7]{SB}.
\item Next, the anticommutation characterization of $\mtxh{2}$ as in Proposition \ref{p:4anticomm} can be also viewed through a group-theoretic lense as follows. Consider the group
$$\Gamma=\left\langle
g_0,\dots,g_5\mid 
g_\ell^2=1\text{ for }\ell\ge0,\ 
g_0g_\ell=g_\ell g_0\text{ for }\ell\ge1,\ 
g_\ell g_m=g_0g_mg_\ell\text{ for }0<\ell<m
\right\rangle.$$
It is easy to see that $\Gamma$ is a finite group, and $g_0$ is central in $\Gamma$. Complex irreducible representations $\pi$ of $\Gamma$ are then distinguished by whether $g_0$ attains $1$ or $-1$ in them. If $\pi(g_0)=1$, then $\pi$ is 1-dimensional; if $\pi(g_1)=-1$, then $\pi$ is 4-dimensional and unique by Proposition \ref{p:4anticomm}.
By the Gowers-Hatami theorem \cite[Theorem 6.9]{GH} (more precisely, its version for the $\ket{\psi}$-induced norm \cite{vidick}), the maps $g_\ell\mapsto X_\ell'$ and $g_\ell\mapsto X_\ell''$ are $\cO(\sqrt{\varepsilon})$ close to compressions of representations of $\Gamma$. 
We may assume that these representations are unitarily equivalent to direct powers of the 4-dimensional representation of $\Gamma$ because 1-dimensional representations can be discarded if $\varepsilon$ is small enough (as they are far away from $X_\ell',X_\ell''$ since $\chsh\le 2$ on scalar observables).
Hence, there are isometries $V_A,V_B$ such that $X_\ell,X_\ell''$ agree with $V_A^*(X_\ell\otimes \Id)V_A,V_B^*(X_\ell\otimes \Id)V_B$ on $\ket{\psi}$ up to $\cO(\sqrt{\varepsilon})$. 
\item Finally, a standard norm estimate on approximate eigenvectors of $(X_1\otimes X_1+X_2\otimes X_2-X_3\otimes X_3+X_4\otimes X_4)\otimes \Id$ and Lemma \ref{l:eigvec} then show that $\ket{\psi}$ agrees with $(V_A\otimes V_B)^*(\ket{\phi_4}\otimes \ket{\aux})$ for some $\ket{\aux}$ up to  $\cO(\sqrt{\varepsilon})$.
\end{enumerate}

\end{remark}

\subsection{Generating quaternions by few projections}\label{sec:quat_prelim}

One aspect of this paper pertains to features of strategies whose measurements generate $\mtxh{n}$ as a real algebra. In this subsection we show that small collections of projections can generate $\mtxh{n}$.
While it is well known that $\mtxr{n}$ and $\mtxc{n}$ for $n\ge 2$ can be generated by three projections as real algebras (e.g., \cite{davis55}), the quaternion counterpart of this statement is addressed by the following proposition (which may be of independent interest in the study of real operator algebras).

\begin{proposition}\label{p:quat}
As a real unital algebra, 
\begin{itemize}
\item[(a)] $\mtxh{n}$ for $n\in\{2,3\}$ is generated by four projections, but not by three projections;
\item[(b)] $\mtxh{n}$ for $n\ge 4$ is generated by three projections, but not by two projections.
\end{itemize}
\end{proposition}

\begin{proof}
First, we record the following version of Jordan's lemma: for any projections $P_1,P_2\in\mtxh{n}$ there exists a unitary $U\in\mtxh{n}$ such that the pair $(UP_1U^*,UP_2U^*)$ is a direct sum of pairs of $1\times 1$ or $2\times 2$ \emph{real} matrices. Indeed,
by \cite[Corollary 6.2]{zhang97}, the hermitian matrix $P_1+P_2$ can be diagonalized as $P_1+P_2=V\Lambda V^*$ where $V\in\mtxh{n}$ is unitary and $\Lambda\in\mtxr{n}$ is diagonal. For a column $v\in\HH^n$ of $V$ let $\lambda\in\R$ be the corresponding eigenvalue of $P_1+P_2$. 
Then $\{v,P_1v\}$ spans a joint invariant subspace for $P_1,P_2$ on which they are represented by a pair of real matrices.

(a) A direct calculation shows that the projections
$$
\begin{pmatrix}1&0\\0&0\end{pmatrix},\ 
\frac12\begin{pmatrix}1&1\\1&1\end{pmatrix},\ 
\frac12\begin{pmatrix}1&-i\\i&1\end{pmatrix},\ 
\frac12\begin{pmatrix}1&-j\\j&1\end{pmatrix}
$$
generate $\mtxh{2}$.
Next, we show that projections $P_1,P_2\in\mtxh{2}$ and a hermitian $Q\in\mtxh{2}$ cannot generate $\mtxh{2}$. By the first paragraph it suffices to assume that $P_1,P_2\in\mtxr{2}$, and $$Q=\begin{pmatrix}\alpha_1&q\\q^*&\alpha_2
\end{pmatrix},\qquad \alpha_\ell\in\R,\ q\in\HH.$$
By \cite[Lemma 2.1]{zhang97}, there is $u\in\HH$ of norm 1 such that $uqu^*\in\C$. Then the unitary $U=u\Id_2\in\mtxh{2}$ commutes with $P_1$ and $P_2$, and $U QU^*\in\mtxc{2}$. Hence, $\Alg_\R(P_1,P_2,Q)$ is isomorphic to a real subalgebra of $\mtxc{2}$, and thus distinct from $\mtxh{2}$.

Let $n=3$. A direct calculation shows that the projections
$$
\begin{pmatrix}1&0&0\\0&0&0\\0&0&0\end{pmatrix},\ 
\frac12\begin{pmatrix}1&1&0\\1&1&0\\0&0&0\end{pmatrix},\ 
\frac13v_1v_1^*,\ 
\frac13v_2v_2^*,\qquad
\text{where}\quad 
v_1=\begin{pmatrix}1\\i\\1\end{pmatrix},\ 
v_2=\begin{pmatrix}1\\j\\1\end{pmatrix}
$$
generate $\mtxh{3}$. 
Next, we show that projections $P_1,P_2,P_3\in\mtxh{3}$ cannot generate $\mtxh{3}$. By the first paragraph it suffices to assume that $P_1,P_2\in\mtxr{3}$ and $(P_1)_{13}=(P_1)_{23}=(P_2)_{13}=(P_2)_{23}=0$. Let
$$P_3=\begin{pmatrix}\alpha_1&q_1&q_2\\q_1^*&\alpha_2&q_3\\q_2^*&q_3^*&\alpha_3
\end{pmatrix},\qquad \alpha_\ell\in\R,\ q_\ell\in\HH.$$
As in the previous paragraph, by \cite[Lemma 2.1]{zhang97} there is $u_1\in\HH$ of norm 1 such that $u_1q_1u_1^*\in\C$. Without loss of generality, let $q_2\neq0$ (the case $q_3\neq0$ is analogous, and the case $q_2=q_3=0$ reduces to the $n=2$ case), and denote $u_2=\frac{u_1q}{|u_1q|}\in\HH$. Then $u_1q_2u_2^*\in\R$. Since $P_3$ is a projection,
$$\begin{pmatrix}q_1^*&\alpha_2&q_3\end{pmatrix}\begin{pmatrix}
\alpha_1\\ q_1^*\\q_2^*
\end{pmatrix}=q_1^*,$$
and thus
$$(\alpha_1+\alpha_2)(u_1q_1^*u_1^*)+(u_1q_3u_2^*)(u_1q_2u_2^*)^*=(u_1q_1^*u_1^*).$$
Since $u_1q_1^*u_1^*\in\C$ and $u_1q_2u_2^*\in\R$, it follows that $u_1q_3u_2^*\in\C$. 
Then the unitary $U=u_1\Id_2\oplus u_2\in\mtxh{3}$ commutes with $P_1$ and $P_2$, and $U P_3U^*\in\mtxc{3}$. Hence, $\Alg_\R(P_1,P_2,P_3)$ is isomorphic to a real subalgebra of $\mtxc{3}$, and thus distinct from $\mtxh{3}$.

(b) Let $n\ge 4$. By the first paragraph of the proof, $\mtxh{n}$ cannot be generated by two projections. To see that three projections suffice, let $\theta_\ell=\frac{\pi}{\ell+2}$ for $\ell\in\N$, and define
$$
P_1=\bigoplus_{\ell=1}^{\lfloor n/2\rfloor}\begin{pmatrix}
1&0\\0&0
\end{pmatrix}\oplus 0^{n-2\lfloor n/2\rfloor},\qquad 
P_2=\bigoplus_{\ell=1}^{\lfloor n/2\rfloor}\frac12
\begin{pmatrix}
(\cos \theta_\ell)^2&\cos\theta_\ell\sin\theta_\ell\\
\cos\theta_\ell\sin\theta_\ell&(\sin \theta_\ell)^2
\end{pmatrix}\oplus 0^{n-2\lfloor n/2\rfloor}.
$$
Then $P_1,P_2\in\mtxr{n}$ are projections, and
\begin{equation}\label{e:realblocks}
\Alg_\R(P_1,P_2)=\left(\bigoplus_{\ell=1}^{\lfloor n/2\rfloor} \mtxr{2}\right)\oplus\R^{n-2\lfloor n/2\rfloor}
\end{equation}
by, e.g., \cite[proof of Theorem 1]{davis55}.
For the third projection, we define $P_3=\frac1n vv^*\in\mtxh{n}$ where $v\in\HH^n$ is given as
$$v^*=\begin{pmatrix}
1&i&1&j&1&1&\cdots&1
\end{pmatrix}.$$
For $1\le \ell,m\le n$ let $E_{\ell m}$ denote the standard matrix units (with $(\ell,m)$-entry equal to one, and zeros elsewhere).
By Eq.~\eqref{e:realblocks} we have
\begin{equation}\label{e:topblocks}
\sum_{m=1}^4\left(A_m\oplus 0^{n-2}\right)\cdot P_3\cdot \left(0^{2\ell}\oplus B_m \oplus 0^{n-2-2\ell}\right)
\in \Alg_\R(P_1,P_2,P_3)
\end{equation}
for all $\ell\in\{0,\dots,\lfloor \frac{n}{2}\rfloor-1\}$ and $A_m,B_m\in\mtxr{2}$.
Since the first row of $nP_3$ is $v^*$, it follows by Eq.~\eqref{e:topblocks} that
$$iE_{11},\ jE_{13},\ E_{1\ell} \in \Alg_\R(P_1,P_2,P_3)$$
for all $1\le \ell\le n$.
Since $\Alg_\R(P_1,P_2,P_3)$ is closed under the conjugate transpose, we also have $E_{\ell1}\in \Alg_\R(P_1,P_2,P_3)$ for all $1\le \ell\le n$, and therefore $E_{\ell m}\in \Alg_\R(P_1,P_2,P_3)$ for all $1\le\ell,m\le n$. Thus, $\mtxr{n}\subset\Alg_\R(P_1,P_2,P_3)$. Since we also have $iE_{11},jE_{13}\in \Alg_\R(P_1,P_2,P_3)$, it follows that $i\Id,j\Id\in\Alg_\R(P_1,P_2,P_3)$. Thus, $\Alg_\R(P_1,P_2,P_3)=\mtxh{n}$.
\end{proof}

\section*{Acknowledgements}
R.C. and L.M. were in part supported by Villum Fonden via a Villum Young Investigator grant (No. 37532). R.C. acknowledges support from the Guangdong Provincial Quantum Science Strategic Initiative (Grant No.~GDZX2403001, GDZX2403008). L.M. additionally acknowledges support from the ERC (QInteract, Grant Agreement No. 101078107). J.V. was supported by the NSF (Grant No. DMS-2348720), and by the Marsden Fund MFP-UOA2528 from Government funding, administered by the Royal Society Te Ap\={a}rangi.

\bibliographystyle{alpha}
\bibliography{quantum}

\end{document}